\documentclass[letterpaper, 10 pt]{ieeetran}
\IEEEoverridecommandlockouts
\usepackage[T1]{fontenc}
\usepackage[utf8]{inputenc}
\usepackage{graphicx}
\usepackage{amsmath,mathtools,booktabs}
\usepackage{amssymb}
\usepackage{url}
\usepackage[dvipsnames]{xcolor}
\usepackage{multirow}

\newcommand{\react}[1]{\ensuremath{\xrightharpoonup{\hbox{\scriptsize $#1$}}}}
\newcommand{\revreact}[2]{\ensuremath{ \xrightleftharpoons[\hbox{\scriptsize$#2$}]{\hbox{\scriptsize$#1$}}}}

\newcommand{\numreac}{m} 
\newcommand{\numconc}{n} 

\usepackage[caption=false,labelfont={footnotesize,bf,rm},textfont={footnotesize,rm},format=hang]{subfig}

\newtheorem{assumption}{Assumption}
\newtheorem{lemma}{Lemma}
\newtheorem{proposition}{Proposition}
\newtheorem{theorem}{Theorem}
\newtheorem{definition}{Definition}
\newtheorem{corollary}{Corollary}

\newtheorem{remark}{Remark}
\newtheorem{example}{Example}
\newtheorem{procedure}{Procedure}
\usepackage{algorithm}
\usepackage[noend]{algpseudocode}
\graphicspath{{./image/}}

\title{\LARGE \bf Structural polyhedral stability of a biochemical network\\
is equivalent to finiteness of the associated generalised Petri net
}

\author{Franco Blanchini$^a$, Carlos Andr\'es Devia$^b$ and Giulia Giordano$^c$
\thanks{The work of C.A.D. was supported by the Delft Technology Fellowship grant awarded to G.G., who also acknowledges the support of the Strategic Grant MOSES at the University of Trento.}
\thanks{$^a$ Dipartimento di Scienze Matematiche, Informatiche e Fisiche, Universit\`a degli Studi di Udine, Via delle Scienze 206, 33100 Udine, Italy. {\tt \small blanchini@uniud.it}}
\thanks{$^b$ Delft Center for Systems and Control, Delft University of Technology, Mekelweg 2, 2628 CD Delft, The Netherlands. {\tt\small c.a.deviapinzon@tudelft.nl}}
\thanks{$^c$ Department of Industrial Engineering, University of Trento, Via Sommarive 9, 38123 Trento, Italy. {\tt\small giulia.giordano@unitn.it}}
}

\begin{document}
\maketitle

\begin{abstract}
We consider biochemical systems associated with a generalised class of Petri nets with possibly negative token numbers. We show that the existence of a structural polyhedral Lyapunov function for the biochemical system is equivalent to the boundedness of the associated Petri net evolution or, equivalently, to the finiteness of the number of states reachable from each initial condition.
For networks that do not admit a polyhedral Lyapunov function,
we investigate whether it is possible to enforce polyhedral structural stability 
by applying a strong negative feedback
on some \emph{pinned} nodes: in terms of the Petri net, this is equivalent to turning pinned nodes into \emph{black holes} that clear any positive or negative incoming token. If such nodes are chosen so that the transformed Petri net has bounded discrete trajectories, then there exists a stabilising pinning control: the biochemical network becomes Lyapunov stable if a sufficiently strong local negative feedback is applied to the pinned nodes. These results allow us to structurally identify the critical nodes to be locally controlled so as to ensure the stability of the whole network.
\end{abstract}

\section{Introduction and Motivation}

 Structural analysis investigates how several systems often encountered in nature
 enjoy important properties in view of their interconnection structure, regardless 
of parameter values.
Here, we consider structural stability and \emph{stabilisation} 
of biochemical networks \cite{BG2021survey,Clarke80,Feinberg2019}, adopting piecewise-linear Lyapunov functions
\cite{BlanchiniGiordano2014,BlanchiniGiordanoCDC2015,BlanchiniGiordano2017}, which -- along with the complementary piecewise-linear-in-rate Lyapunov functions
\cite{AlrAng13,AlrAng16,AlrAngSon19,BlanchiniGiordanoCDC2015} that can be seen as their dual \cite{BG2021} -- have proven effective in the 
stability analysis of chemical reaction networks. A recent contribution \cite{AlrAngSon20} shows that this type of functions can be very useful to detect, more in general, non-oscillatory behaviours.

Chemical reaction networks have been often analysed resorting to discrete-event frameworks, employing for instance Petri nets \cite{AngDeLSon2007, Soliman2012}: a chemical reaction is seen as a process that assembles the needed number of reactant molecules and releases the proper number of product molecules.

In this paper, we show how the existence of a piecewise-linear Lyapunov function for the large class of \emph{unitary} (bio)chemical networks can be interpreted as the boundedness of the evolution of a suitable \emph{generalised} Petri net, which can have both positive and negative token numbers.

This generalisation of Petri nets has been widely investigated in the literature, e.g. under the name of \emph{lending Petri net} \cite{Bartoletti201575}, and \emph{negative tokens} \cite{MuYam90} have been also called \emph{anti-tokens} \cite{Gerogiannis1998133,Sokolov2007197} or \emph{debit tokens} \cite{Bartoletti20151,KohlerBussmeier2008329}.
However, to the best of the authors' knowledge, this is the first time that such a concept is associated with the stability of biochemical networks.

We also consider a \emph{structural} pinning control problem. Pinning some nodes means applying a strong feedback to these nodes with the goal of controlling the whole network. 

Pinning control has been extensively investigated in past years \cite{Li2004,Wang2014103}, with one of the main driving questions being how many nodes to pin and which ones \cite{Orouskhani2016}.
The approach has been used to address network control problems ranging from 
asymptotic convergence  \cite{DeLellis20133033} and noise rejection \cite{BRB2019}
to consensus  \cite{Chen20091215} and synchronisation \cite{Porfiri20083100}. Pinning techniques have been applied in different areas such as circuits \cite{DeLellis20133033}, power grids \cite{Orouskhani2016}, 
protein networks and gene regulatory networks \cite{BRB2019}.
Pinning control of a (bio)chemical reaction network system can be seen as equivalent to the conversion of the pinned nodes into 
\emph{black holes}, which swallow any incoming token (either positive or negative),
in the associated generalised Petri net. 

The main contributions of this paper are summarised next.

\begin{itemize}
\item We formulate the structural stability problem for biochemical reaction networks (Section \ref{setup}) and we associate a biochemical network with a Generalised Petri Net (GPN), with possibly negative tokens (Section \ref{polybounded}).
\item The GPN is fully determined by the network structure, and it does not depend on the (monotonic) reaction rate functions.
\item  The boundedness of all possible evolutions of the GPN is equivalent to the existence of a polyhedral Lyapunov function (PLF) for the biochemical system; such a PLF can be computed based on the efficient numerical procedure proposed in \cite{BlanchiniGiordano2014} (Section \ref{subsec:petri}).
\item The equivalence with the GPN suggests more efficient stopping criteria for the numerical procedure (Section \ref{subsec:stop}).
\item For networks that do not admit a PLF, we study how to convert some nodes into \emph{black holes} that swallow any incoming token (Section \ref{subsec:blackholes}), so as to ensure the stability of the network.
\item We show that converting some nodes into black holes is equivalent to \emph{pinning} them, i.e. virtually fixing their state variables to imposed values (Section \ref{sec:pinningcontrol}); after pinning appropriately chosen nodes, the network can admit a PLF.
\item We illustrate our results by assessing the structural stability, possibly after pinning suitably chosen nodes, of some examples from the biochemical literature, including transcription and translation models (Section \ref{sec:applications}).
\end{itemize}

\section{Problem Formulation}\label{setup}
Consider the general class of biochemical systems
\begin{equation}\label{sys}
\dot x(t) = S g(x(t)) + g_0,
\end{equation}
where the state vector $x(t) \in \mathbb{R}^n_+$ includes the concentrations of the involved biochemical species, $S \in \mathbb{Z}^{n \times m}$ is the stoichiometric matrix, the vector function $g: \mathbb{R}^n_+ \to \mathbb{R}^m$ represents reaction rates, and $g_0 \in \mathbb{R}^n_+$ is a vector of constant influxes. We make the following standing assumptions.
\begin{assumption}\label{as_eq}
System \eqref{sys} admits the equilibrium $\bar x \in \mathbb{R}^n_+$, such that
\begin{equation}
\label{eq}
0 = S g(\bar x) + g_0.
\end{equation}
\end{assumption}
\begin{assumption}\label{unit}
The network is \emph{unitary} \cite{BlanchiniGiordano2014}, namely, each of the entries of matrix $S$ is either $1$, $0$ or $-1$.
\end{assumption}

\begin{assumption}\label{incr}
Function $g_k(x)$ is nonnegative and strictly monotonic (either increasing or decreasing) in each of its arguments; it depends on variable $x_i$ if and only if $S_{ik}=-1$,
and it is zero if and only if one of its arguments is zero.
\end{assumption}

\begin{remark}
Unitary networks, as in Assumption \ref{unit}, cannot include multi-molecular reactions, such as $2 X_1+X_2 \react{} X_3$. However, multi-molecular reactions are known to occur in fact as chains of bi-molecular reactions (e.g., $X_1+X_2\react{} X_4$ and $X_1+X_4\react{} X_3$),
which do lead to a unitary network.
Note also that Assumption \ref{incr} rules out autocatalytic reactions (e.g. of the form $X_1 \react{} 2X_1$), for which structural stability could never be guaranteed.
\end{remark}

As shown in \cite{BlanchiniGiordano2014,BlanchiniGiordano2017,BG2021survey},
the variable shift $z(t)=x(t)-\bar x$ allows us to write the system in the equivalent form
\begin{equation}\label{bdc}
\dot z(t) = BD(z(t))C z(t),
\end{equation}
where $D$ is a diagonal matrix with positive diagonal entries,
$$
D(z) =\mbox{diag}\{D_1(z), D_2(z),   \dots ,D_q(z)\}, ~~~~D_k(z) > 0.
$$
Matrices $B$, $D$ and $C$ are derived from \eqref{sys} as follows
\begin{itemize}
\item The diagonal entries of $D$ are related to the absolute values of the $q$ nonzero partial derivatives $\partial g_j/\partial x_i$, with $j \in \{ 1, \dots, m\}$, $i \in \{ 1, \dots, n \}$, arbitrarily ordered.
\item The $k$th column of $B$, $B_k$, corresponding to $D_{kk}$ related to $|\partial g_j/\partial x_i|$, is equal to the column $S_j$ of $S$.
\item The $k$th row of $C$, $C_k^\top$, corresponding to $D_{kk}$ related to $|\partial g_j/\partial x_i|$, has a single nonzero entry, the $i$th, equal to the sign of $\partial g_j/\partial x_i$.
\end{itemize}
The proof \cite{BlanchiniGiordano2017,BG2021survey} relies on the fact that the Jacobian of $Sg(x)$ can be written as $J(x) = B \Delta C$, where the diagonal matrix $\Delta$ includes the absolute value of all partial derivatives, and on the integral formula

{\footnotesize
\begin{eqnarray*}
&&Sg(x)-Sg(\bar x) = \left [\int_0^1~J(\bar x + \sigma(x-\bar x)) d \sigma \right ](x-\bar x)\\
&=&  B \left [\int_0^1~\Delta(\bar x + \sigma(x-\bar x)) d \sigma \right ] C(x-\bar x) =
B D(x)C(x-\bar x).
\end{eqnarray*}}
\emph{Structural} stability (which needs to hold regardless of the numerical values and functional expressions within matrix $D$, \cite{BlanchiniGiordano2014,BlanchiniGiordano2017,BG2021survey}) can be studied by absorbing system \eqref{bdc} in a differential inclusion
 \begin{equation}\label{diffinc}
\dot z(t) = BD(t)C z(t),
\end{equation}
where $D(t)$ is a positive definite diagonal matrix of size $q$. 

\begin{figure}[tb]
\centering
\subfloat{\includegraphics[width=4cm]{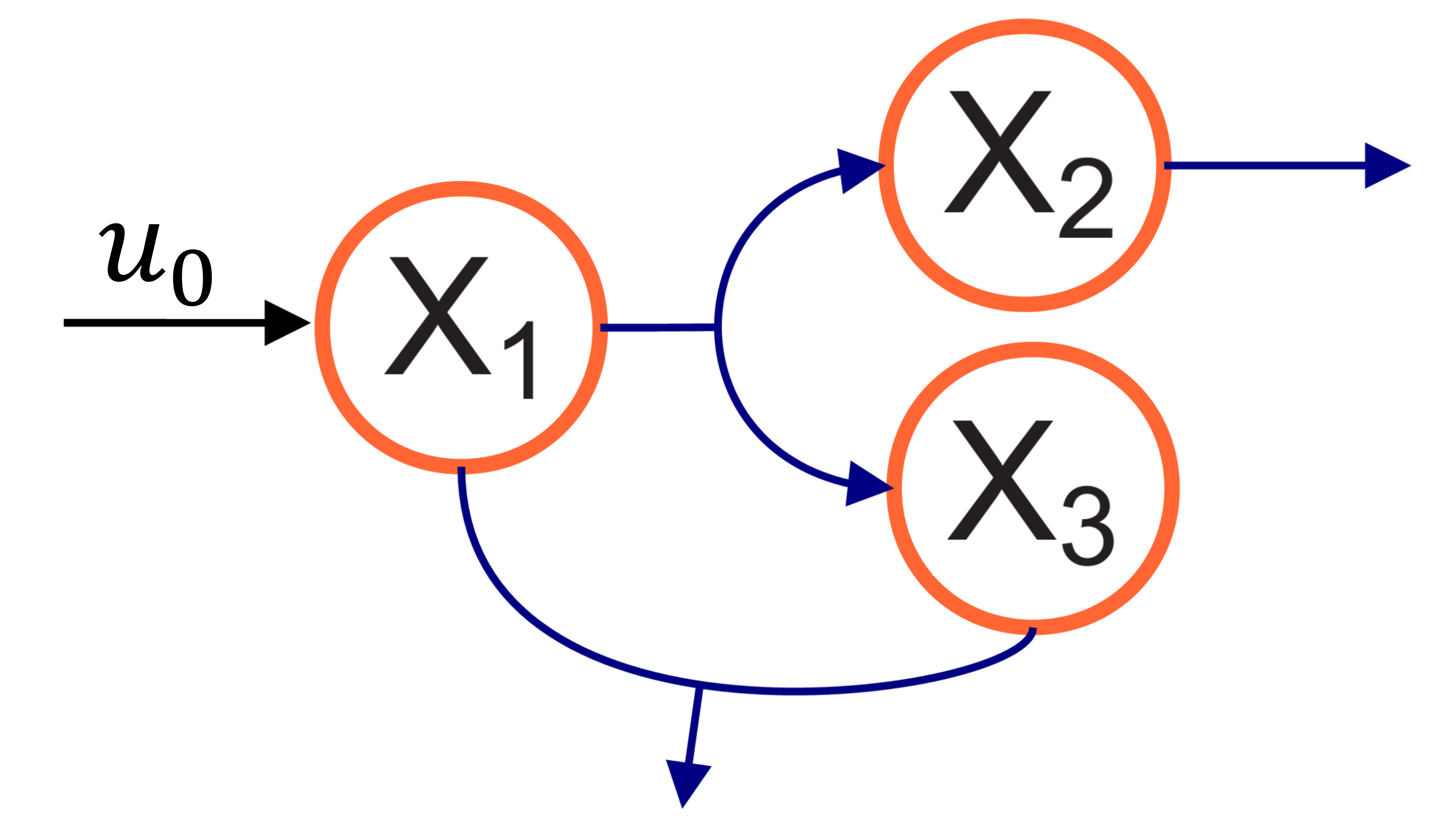}} \,
\subfloat{\includegraphics[width=4cm]{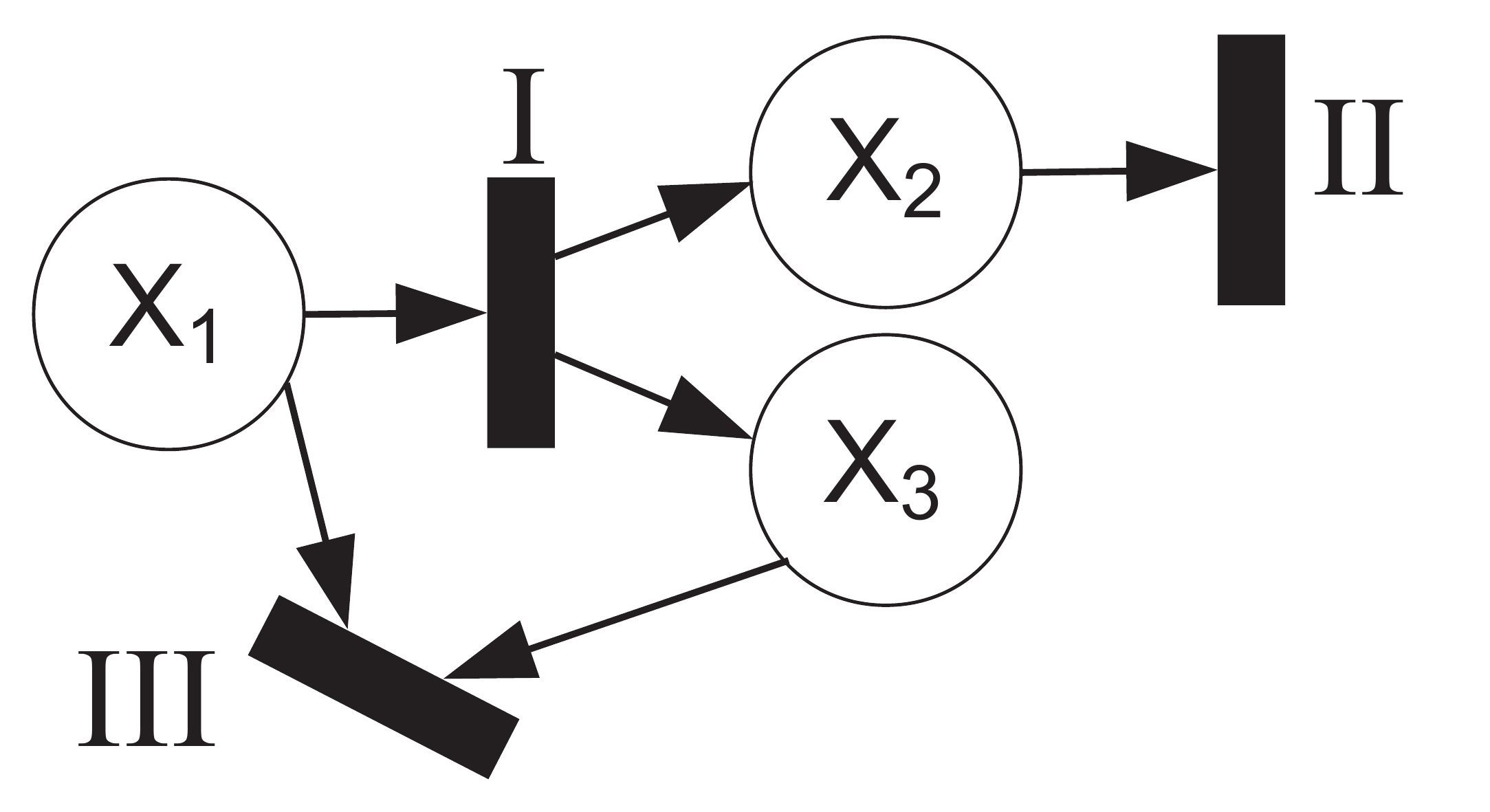}}
\caption[Graph and discrete-event representation.]{Graph and GPN representation of the system in Example \ref{ex:Telemann}.}
\label{fig:Petri}
\end{figure}

\begin{example}\label{ex:Telemann}
Consider the biochemical reaction network $\emptyset \react{u_1} X_1$, $X_1 \react{g_1} X_2 + X_3$, $X_1 + X_3 \react{g_{13}} \emptyset$, $X_2 \react{g_2} \emptyset$, admitting the graph representation in Fig.~\ref{fig:Petri}, left. The system evolution is described by the differential equations
\begin{equation*}
\begin{cases}
\dot x_1 = u_1 -g_1(x_1)-g_{13}(x_1,x_3)\\
\dot x_2 = g_1(x_1)-g_2(x_2)\\
\dot x_3 = g_1(x_1)-g_{13}(x_1,x_3)
\end{cases}
\end{equation*}
which can be recast in the form \eqref{bdc} with
\begin{equation*}
B = \left[ \begin{smallmatrix}
-1 & ~~0 & -1 & -1\\
~~1 & -1 & ~~0 & ~~0\\
~~1 & ~~0 & -1 & -1
\end{smallmatrix} \right] \quad \mbox{and} \quad
C = \left[ \begin{smallmatrix}
1 & 0 & 0\\
0 & 1 & 0\\
0 & 0 & 1\\
1 & 0 & 0
\end{smallmatrix}\right].
\end{equation*}
\end{example}

\begin{remark}
For unitary networks, satisfying Assumption \ref{unit}, $C_i^\top B_i=-1$ for all $i=1,\dots,q$.
\end{remark}

Henceforth we work under an additional, mild assumption.
\begin{assumption} \label{largebound}
There exist {\bf unknown} bounds $D^-_k>0$ and $D^+_k>0$ such that
\begin{equation}
D^-_k \leq D_k \leq D^+_k, \quad k=1,\dots,q. \label{largeboundeq}
\end{equation}
\end{assumption}

\begin{definition} \label{epsi}
The system \eqref{bdc} under Assumptions \ref{as_eq}-\ref{largebound}
is \emph{structurally stable} if the differential inclusion \eqref{diffinc} is Lyapunov stable (possibly marginally), while it is \emph{structurally asymptotically stable} if \eqref{diffinc}
is asymptotically stable.
\end{definition}

\begin{remark}
The purely technical Assumption \ref{largebound} does not change the structural nature
of our investigation. The upper bound $D^+_k$  ensures compactness and can be arbitrarily large.
 The arbitrarily small $D^-_k$ ensures $D_k$ to be bounded away from zero
(cf. the $\epsilon$-perturbation in \cite[Definition 2]{AlrAngSon20}).
For instance, the scalar system $\dot z(t) =-D(t) z(t)$ is not necessarily asymptotically stable with the weaker bound $D>0$ (e.g., if $D(t)=e^{-t}$, $\dot z =-D(t) z(t)$ does not converge to $0$), while stability is asymptotic if $D(t)\geq D^->0$.
\end{remark}

\section{Polyhedral Lyapunov Functions and Generalised Petri Net Boundedness} 
\label{polybounded}
\begin{definition}\label{def:Lyap}
Given an uncertain dynamical system $$\dot x=f(x,w),~~~w \in {\cal W},$$
where ${\cal W}$ is a closed set, the positively homogeneous convex function
$V(x)$ is a \emph{Lyapunov function (LF)} for the system $\dot x=f(x,w)$
if, for some $\beta\geq0$, the (generalised) Lyapunov derivative 
$$ D^+V(x,w)=\limsup_{h\rightarrow 0}\frac{V(x+hf(x,w))-V(x))}{h} \leq - \beta V(x) $$ 
for all $x$ and $w \in {\cal W}$.
The LF is \emph{weak} if the inequality holds for $\beta=0$, \emph{strong} if $\beta>0$.
The LF $V(x)$ is  \emph{polyhedral (PLF)} if it can be written as
\begin{equation}\label{polyfuncF}
V(x) = \|Fx\|_\infty,
\end{equation}
where matrix $F$ has full column rank, or
\begin{equation}\label{polyfuncX}
V(x) = \min \{\|p\|_1:~x=Xp\}, 
\end{equation}
where matrix $X$ has full row rank.
\emph{Polyhedral stability} means that the system admits a PLF.
\end{definition}

The procedure in \cite{BlanchiniGiordano2014} to generate a PLF for the system
associates 
the original differential inclusion \eqref{diffinc}
with a discrete difference inclusion 
\begin{equation}
y(k+1) = \Phi(k) y(k),~~\Phi(k) \in \mathcal{F},\label{discdiffinc_sw}
\end{equation} 
where
\begin{equation}\label{family_M}
 \mathcal{F} = \{ \Phi_i : \Phi_i = I +  B_i C_i^\top, \,\, i=1,\dots,q\}.
\end{equation} 
The procedure iterates over polyhedral sets,
starting from the unit ball of the $1$-norm: $\mathcal{Y}^0 = \mbox{conv}\{[-I~I]\}$, where $\mbox{conv}$ denotes the convex hull. 
\begin{procedure} 
\label{proc}~\cite{BlanchiniGiordano2014}
\begin{enumerate}
\item $Y^0=I$, $\mathcal{Y}^0 = \mbox{conv}\{[-Y^0~Y^0]\}$;
\item $Y^{k+1} := [\Phi_1~\Phi_2~\dots~\Phi_q]Y^{k}$; 
\item     $\mathcal{Y}^{k+1}:=\mbox{conv}\{[ -Y^{k+1}~ Y^{k+1}]\}$,
\item \label{incomplete} IF $\mathcal{Y}^{k+1} = \mathcal{Y}^{k}$, set $\bar{\mathcal{Y}} = \mathcal{Y}^{k}$ and STOP;
ELSE go to step $2$.
\end{enumerate}
\end{procedure}

For a practical implementation, further stopping conditions should be added before the ELSE statement at Step \ref{incomplete} of the procedure; otherwise, as currently stated, the procedure {\em might fail to stop}. 
Convergence issues are one of the aspects we will investigate: stopping criteria with a negative outcome (i.e., no structural PLF exists) will be discussed in Section \ref{subsec:stop} for numerical purposes.

If Procedure \ref{proc} stops, the polytope $\bar {\mathcal{Y}}=\mathcal{Y}^{k}$,
with vertices $X=[ -Y^{k}~ Y^{k}]$,
is the unit ball of a PLF as in \eqref{polyfuncX}. If we apply the same procedure to the dual system $\dot z=C^\top D B^\top z$,
by considering $\Phi_i^\top$, under convergence assumptions,
we obtain the PLF in the dual form in \eqref{polyfuncF} with $F=Y^\top$, where $Y=[ -Y^{k}~ Y^{k}]$.
The efficient implementation of the procedure requires removing the redundant 
columns at each iteration $k$; see \cite{BlanchiniGiordano2014} for details.
The sequence $\mathcal{Y}^{k}$ being formed by integer vectors drastically improves computability and provides efficient stopping criteria when the procedure fails to converge.

As proven in \cite{BlanchiniGiordano2014}, the stability of the differential inclusion \eqref{diffinc} is equivalent to the stability  of \eqref{discdiffinc_sw};
if \eqref{diffinc} admits a (weak) PLF, then it is marginally stable, which implies that system \eqref{bdc}, with $D(z)$ continuous, is stable,
and is asymptotically stable if and only if its Jacobian $BDC$ is structurally non-singular \cite{BlanchiniGiordano2017}.
Moreover, if $C_i^\top B_i=-1$, the stability of \eqref{diffinc}  is equivalent to the
existence of a PLF for both \eqref{diffinc} and \eqref{discdiffinc_sw}
and also equivalent to the fact that  Procedure \ref{proc} successfully stops in finite time.
If $C_i^\top B_i=-1$, systems \eqref{bdc} and \eqref{diffinc} admit a structural
(weak) Lyapunov function if and only if they admit a (weak) structural PLF.
Hence, the existence of a PLF guarantees stability,
and even asymptotic stability under structural non-singularity assumptions,
as summarised in the following result.
\begin{theorem} \label{nonsing}
\cite{BlanchiniGiordano2017}
Assume that system \eqref{diffinc}, under Assumptions \ref{as_eq}-\ref{largebound},
admits a (weak) PLF.
Then, it is asymptotically stable if and only if matrix $BDC$ is non-singular for all 
possible matrices $D$ satisfying \eqref{largeboundeq}.
\end{theorem}

Asymptotic stability is shown to be exponential in \cite{FBBCG20}.

\begin{remark} \label{rem:nons}
Structural non-singularity is easy to check, as shown in \cite{GCFB16}:
it is equivalent to $\det [-B \hat D C] > 0$ for all matrices $\hat D$ on the vertices of the hyper-rectangle in \eqref{largeboundeq}. 
\end{remark}

\subsection{Generalised Petri net model}\label{subsec:petri}
Procedure \ref{proc} can be interpreted as the evolution of a Generalised Petri Net (GPN), a discrete-event system that, albeit similar to a Petri net,
does not work as a Petri net because the number of tokens at each node can be either positive or negative \cite{MuYam90}. The integer vector $\tau(k)$ denotes the presence of $\tau_\ell(k)$ tokens in position $\ell$ at time $k$.

\begin{example}\label{ex:Telemann_procedure}
The system in Example \ref{ex:Telemann} can be associated with the difference inclusion $y(k+1) = \Phi(k) y(k)$ where, at each step $k$, $\Phi(k)$ is one of the following matrices:
{\small \begin{equation*}
\Phi_1 = \left[ \begin{smallmatrix}
0 & 0 & 0 \\
1 & 1 & 0 \\
1 & 0 & 1 
\end{smallmatrix} \right], \, 
\Phi_2=\left[ \begin{smallmatrix}
1 & 0 & 0 \\
0 & 0 & 0 \\
0 & 0 & 1 
\end{smallmatrix}\right], \,
\Phi_3=\left[ \begin{smallmatrix}
1 & 0 & -1 \\
0 & 1 & ~~0 \\
0 & 0 & ~~0 
\end{smallmatrix}\right], \,
\Phi_4=\left[ \begin{smallmatrix}
~~0 & 0 & 0 \\
~~0 & 1 & 0 \\
-1 & 0 & 1 
\end{smallmatrix} \right].
\end{equation*}}
We start from the unit ball of the $1$-norm, $\mathbb{X}^0 = [-I~~I] = [-v_1~-v_2~-v_3~~v_1~~v_2~~v_3]$, and consider just the positive vertices (the evolution of the others can be immediately obtained, being the opposite). Vertex $v_1 = [1~0~0]^\top$ is transformed into $\Phi_1 v_1 = [0~1~1]^\top = v_4$, $\Phi_2 v_1 = \Phi_3 v_1 =[1~0~0]^\top = v_1$, $\Phi_4 v_1 = [0~0~-1]^\top = -v_3$; vertex $v_2 = [0~1~0]^\top$ is transformed into $\Phi_1 v_2 = \Phi_3 v_2 = \Phi_4 v_2 = [0~1~0]^\top = v_2$, $\Phi_2 v_2 = [0~0~0]^\top$; vertex $v_3 = [0~0~1]^\top$ is transformed into $\Phi_1 v_3 = \Phi_2 v_3 = \Phi_4 v_3 = [0~0~1]^\top = v_3$, $\Phi_3 v_3 = [-1~0~0]^\top = -v_1$.
The sole newly generated vertex is $v_4= [0~1~1]^\top$ (and its opposite). The procedure applied to $v_4$ gives $\Phi_1 v_4 = \Phi_4 v_4 = [0~1~1]^\top = v_4$, $\Phi_2 v_4  =[0~0~1]^\top = v_3$, $\Phi_3 v_1 = [-1~1~0]^\top = v_5$.
Applying the procedure to the only new vertex, $v_5= [-1~1~0]^\top$, gives $\Phi_1 v_5 = [0~0~-1]^\top = -v_3$, $\Phi_2 v_5  =[-1~0~0]^\top = -v_1$, $\Phi_3 v_5 = [-1~1~0]^\top = v_5$, $\Phi_4 v_5 = [0~1~1]^\top = v_4$.
No new vertices are generated at this step, hence the procedure stops successfully: the system admits a PLF having unit ball
$\mbox{conv}\{[-X~X]\}$, with $X = [v_1~~v_2~~v_3~~v_4~~v_5]$.
\end{example}

\begin{figure}[htb!]
\centering
\subfloat[First step, initial marking {$v_1=[1~0~0]^\top$}.]{\includegraphics[width=7cm]{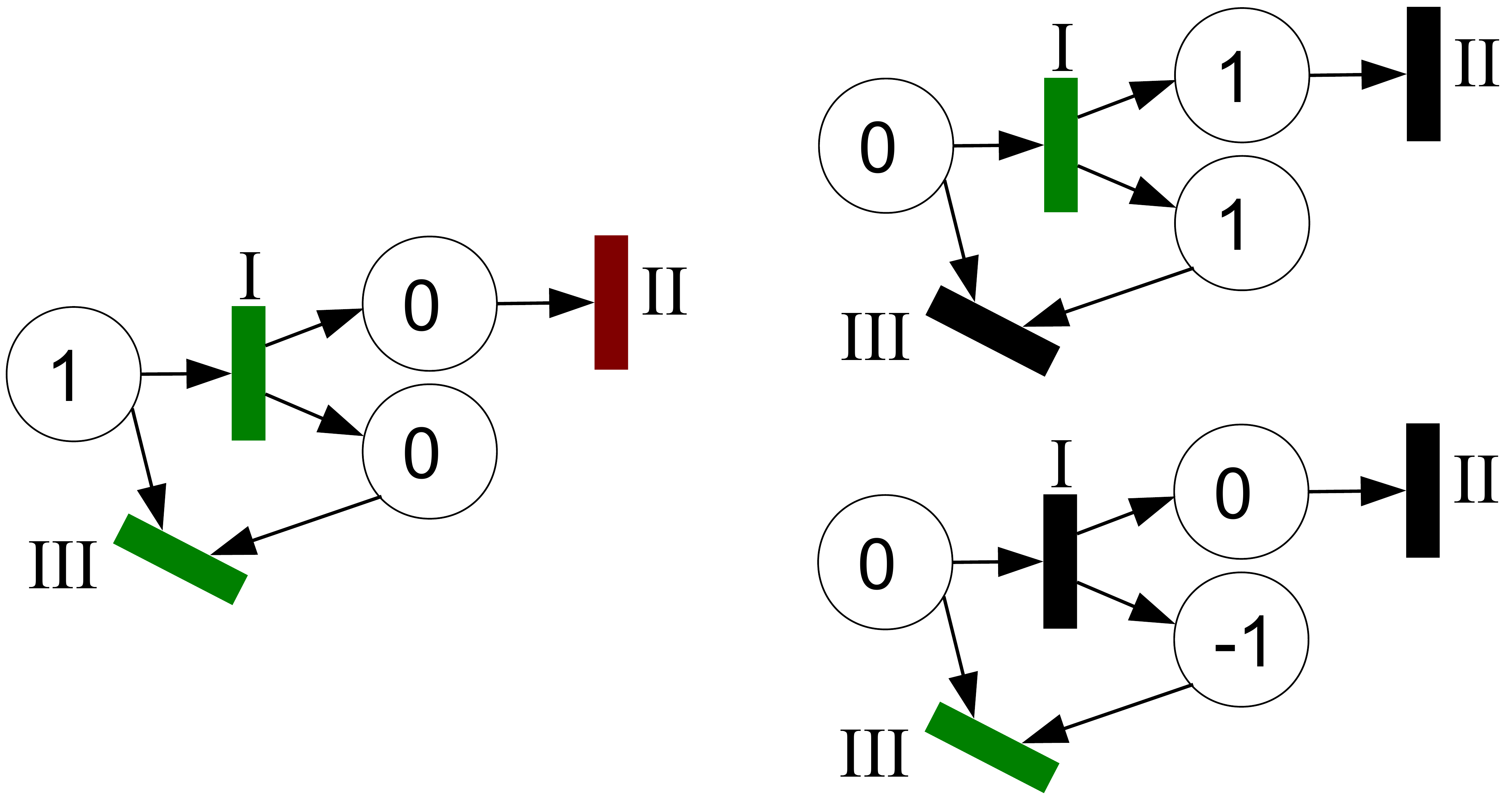}} \\
\subfloat[First step, initial marking {$v_2=[0~1~0]^\top$}.]{\includegraphics[width=7cm]{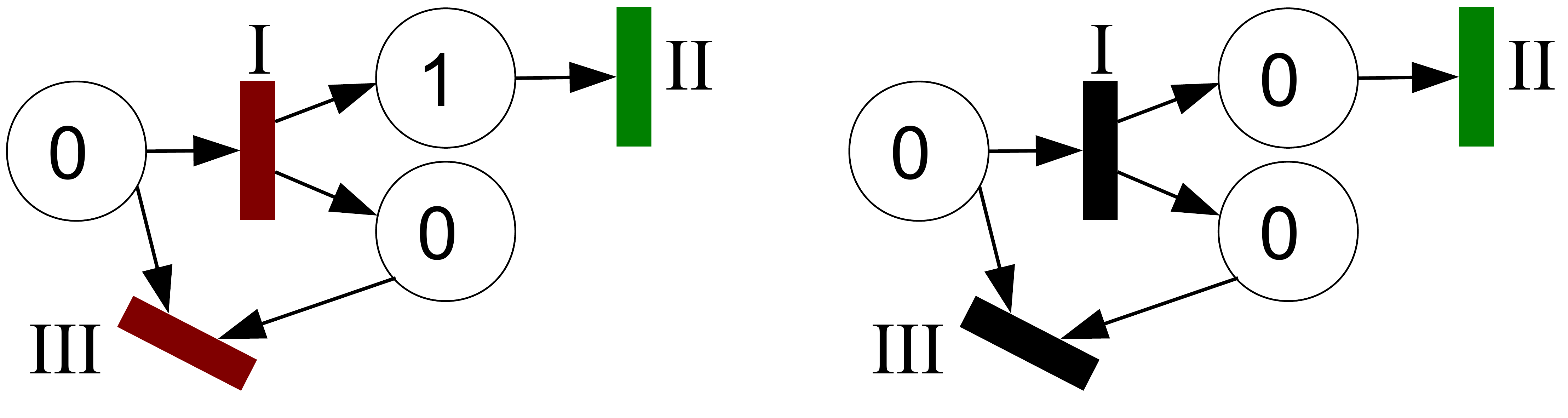}} \\
\subfloat[First step, initial marking {$v_3=[0~0~1]^\top$}.]{\includegraphics[width=7cm]{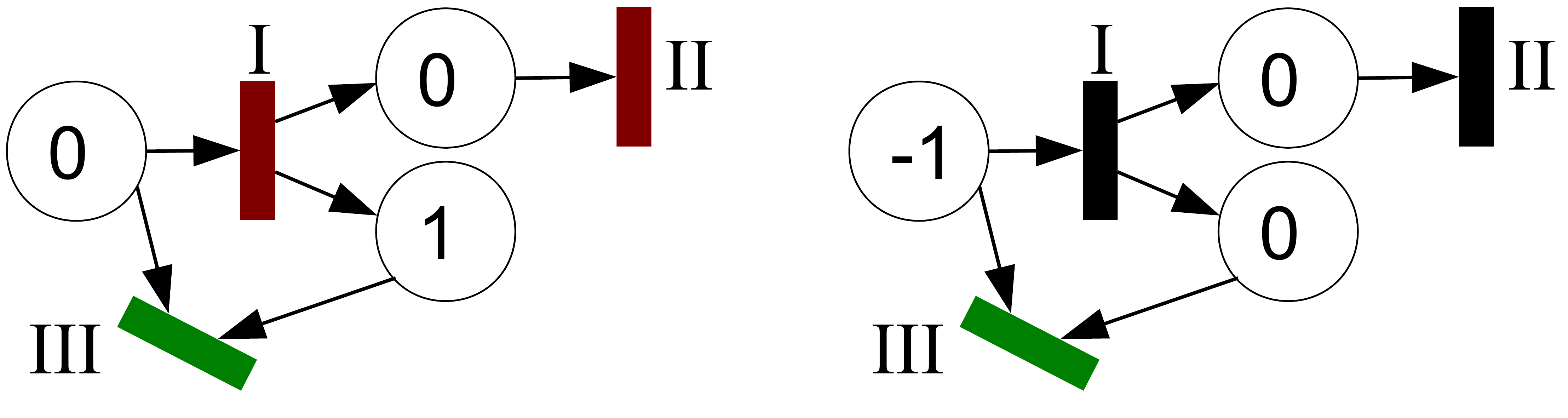}}\\
\subfloat[Second step, initial marking {$v_4=[0~1~1]^\top$}.]{\includegraphics[width=7cm]{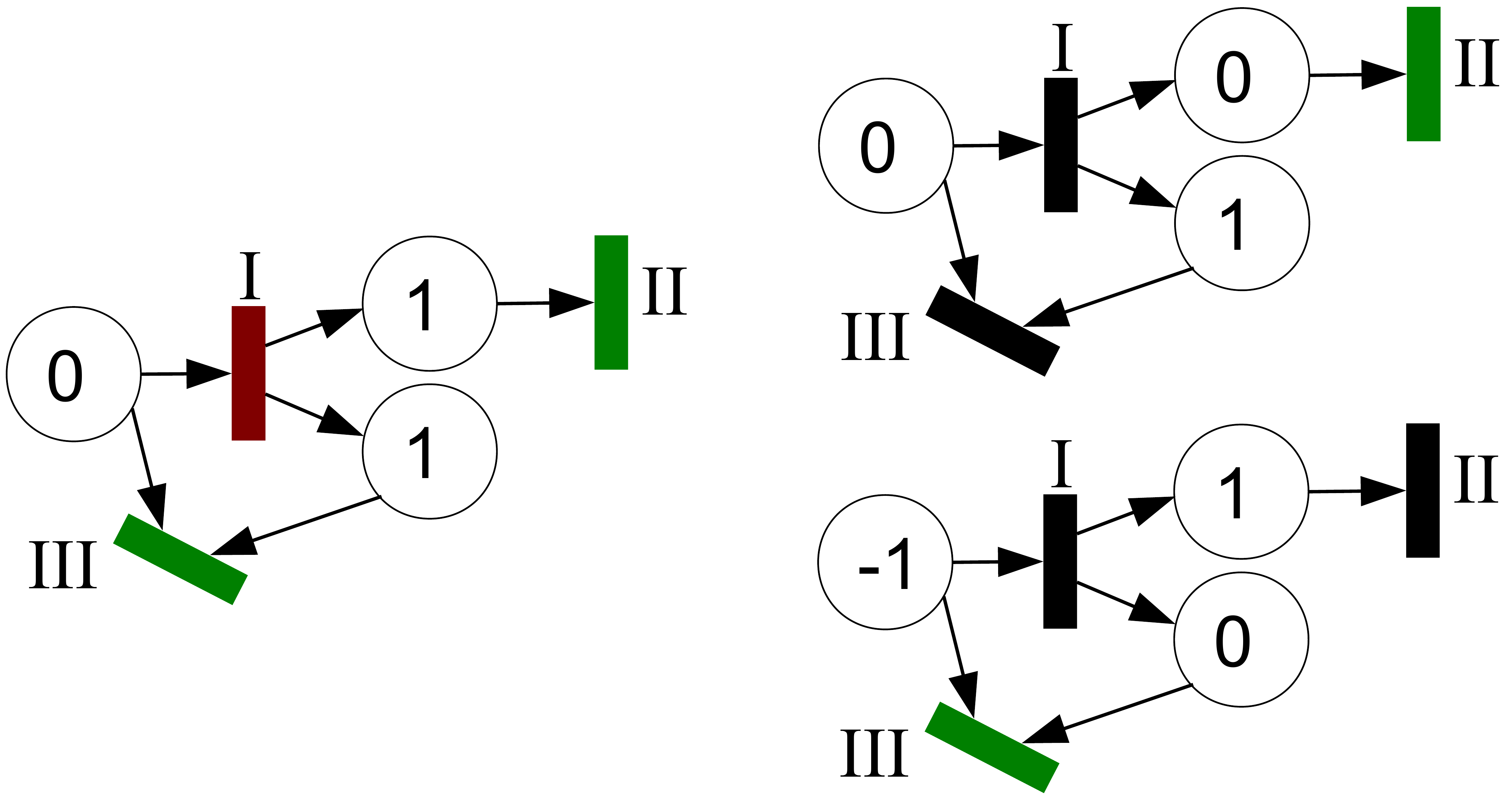}}\\
\subfloat[Third step, initial marking {$v_5=[-1~1~0]^\top$}.]{\includegraphics[width=7cm]{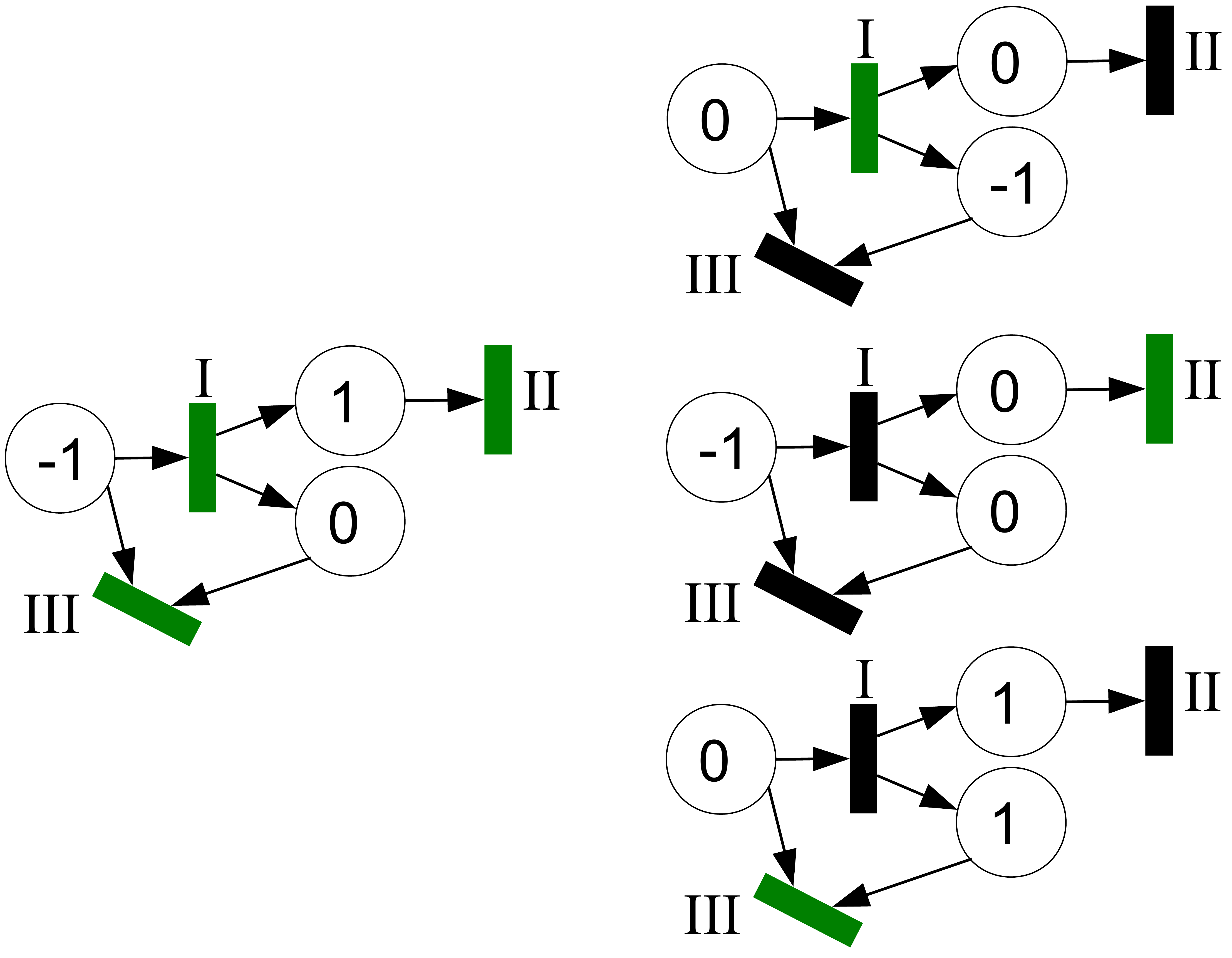}}
\caption{Discrete-event interpretation of the procedure evolution. On the left side, enabled transitions are in green and non-enabled transitions in red. On the right, the activated transition is in green.}
\label{fig:Petri_evolution}
\end{figure}

The evolution of the discrete-time system in the numerical Procedure \ref{proc} can be related to the evolution of a particular discrete-event system, in which the initial conditions represent an \emph{initial marking} $\tau(0)$, assigning an integer number of tokens to each node (associated with each of the chemical species), and a transition (associated with each of the black rectangles in Fig.~\ref{fig:Petri}, right, i.e. with each of the reactions occurring in the network) is \emph{enabled} whenever at least one of the starting nodes of the transition contains a non-zero number of tokens. When either the number of tokens in the starting node is positive, or the starting nodes are two and the number of tokens is non-negative in both (and positive in at least one of them), then the transition takes tokens from the starting node(s) and moves them to the arrival node(s), if explicitly present (otherwise, they simply disappear). When either the starting node is one only and the number of tokens therein is negative, or the starting nodes are two and the number of tokens is non-positive in both (and negative in at least one of them), then the transition takes negative tokens from the starting node(s) and moves them to the arrival node(s), again, if explicitly present.  When a transition is enabled, and performed, all the tokens are moved until one of the starting nodes has zero tokens.
If one token number is negative and the other positive, their effects are superposed. 

\begin{example}\label{Ex:3}
For the system in Examples \ref{ex:Telemann} and \ref{ex:Telemann_procedure}, the GPN evolution is shown in Fig.~\ref{fig:Petri_evolution}: on the left side, the initial marking is illustrated (enabled transitions are in green, non-enabled transitions in red); on the right side, the new marking generated by the action of each of the enabled transitions is illustrated.
Reaction I corresponds to the evolution matrix $\Phi_1$ (and generates the same outcome for the same initial conditions), reaction II corresponds to matrix $\Phi_2$ (and again generates the same outcome for the same initial conditions), while reaction III corresponds to matrices $\Phi_3$ and $\Phi_4$. In this latter case, the outcome of (at least) one of the two evolution matrices is the same as that of the transition related to reaction III in the discrete-event evolution, while the outcome of the other (if different) is always the unchanged input vertex. Interestingly, the evolution is the very same as that of the numerical procedure, and the same new vertices are generated at each iteration.
\end{example}

\begin{proposition}\label{prop}
If $C_i^\top B_i =-1$ $\forall i = 1,\dots,q$, the sequence generated by the discrete-time switching system \eqref{discdiffinc_sw}-\eqref{family_M} uniquely corresponds to the evolution of a GPN discrete-event system if the initial vector $y(0)=\tau(0)$ has integer components.
\end{proposition}
\begin{proof}
To keep the notation simple, we assume that all reactions are functions of at most two variables.
The integer operator $\Phi_h$ corresponds to the $h$th non-zero
derivative, say $\partial g_i/\partial x_j$. Consider the $i$th reaction
$X_{ j } + X_w \react {g_i} X_\ell$.
Then,
$$
\Phi_h =I +B_h C_h^\top  = I + S_i e_j^\top
$$ 
where $S_i$ is the $i$th column of $S$ and $e_j^\top$ is the $j$th canonical row vector.
Hence, $\Phi_h$ has all columns equal to the identity matrix excluding the $j$th. Column $j$ has the same entries as $B_h=S_i$, excluding the diagonal entry $(j,j)$ which is $0$ (see Example \ref{ex:Telemann_procedure}).
Operator $\Phi_h$ applied to any integer vector $\tau$ as $\tau^+ = \Phi_h \tau$, a shorthand notation for $\tau(k+1)=\Phi_h \tau(k)$, corresponds to the following {\bf transition $T_h$} performed on the Petri net:\\
1) $\tau^+_j=0$ (remove $\tau_j$, i.e. all, tokens from node $j$);
\\
2) $\tau^+_\ell=\tau_\ell+  B_{\ell h}  \tau_j =\tau_\ell +\tau_j $ (add $\tau_j$ tokens to node $\ell$);
\\
3) $\tau^+_w=\tau_w+ B_{wh} \tau_j = \tau_w-\tau_j$ (remove $\tau_j$ tokens from $w$).

This defines a one-to-one correspondence between matrix $\Phi_h$ applied on integer vectors and the corresponding transition $T_h$ (such that $\tau^+ = T_h \tau$).
\end{proof}

Henceforth, we denote by $\Phi(h)$, $h=0,1,\dots$, the generic matrix sequence
$\Phi(h) \in \{\Phi_1,\Phi_2, \dots, \Phi_q\} =\mathcal{F}$.
The family of all their products,
\begin{equation}\label{eq:productset}
\Pi(\mathcal{F})\doteq 
\left \{ \prod_{h=0}^K\Phi(h), ~K\geq 0,~~\Phi(h) \in \mathcal{F} \right \},
\end{equation}
is an algebraic semigroup under the multiplication operation, i.e., if both $P_1$ and $P_2$ are in $\Pi(\mathcal{F})$, then also $P_1 P_2 \in \Pi(\mathcal{F})$.

\begin{theorem}\label{th_DES}
If $C_i^\top B_i =-1$ $\forall i = 1,\dots,q$, the following statements are equivalent:\\
(i)  for any initial integer token distribution vector $\tau(0)=\tau_0$,
the set of possible evolutions of the GPN $\tau(k+1) = T_k \tau(k)$,
$$
{\mathcal{R}}(\tau_0) = \{\tau  = T_\ell \circ T_{\ell-1}  \dots  \circ  T_0 ~\tau_0, \ell\geq 0, ~T_i~\mbox{arbitrary}\},
$$
i.e. the reachable set from $\tau_0$, is finite.\\
(ii) The semigroup
$\Pi(\mathcal{F}) $ in \eqref{eq:productset} is finite.\\
(iii) Procedure \ref{proc} stops in finite time.
\end{theorem}
\begin{proof}
%
(i) $\Rightarrow$ (ii):  If the GPN generates a finite number of configurations
given any initial token distribution $\tau(0)$, then
from Proposition \ref{prop}, for any integer $y(0)$,
the number of states reached by
$y(k+1) = \Phi(k) y(k)$, is finite, hence bounded.
Since the generic trajectory given $y(0)$
has the form
$$
y(k) = \Phi(k)\Phi(k-1) \dots\Phi(1) \Phi(0) y(0),
$$
this implies that the products in the set $\Pi(\mathcal{F})$ are uniformly bounded by some constant $\mu>0$: $\|\prod_{h=0}^K\Phi(h)\| \leq \mu$. Then (ii)
follows because all these products are integer matrices.
\\
(ii) $\Rightarrow$ (iii):
Procedure \ref{proc} does stop in finite time,
because all the columns of the matrices $Y^k$ are generated as $\prod_{h=0}^K\Phi(h)Y^0$, which are in a finite number.
\\
(iii) $\Rightarrow$ (i):
If Procedure \ref{proc} stops in finite time, the matrix family $\mathcal{F}$
admits the PLF induced by the final set $\bar{\mathcal{Y}}$.
Hence, for any integer initial condition $y(0)$ the sequence $y(k)$ is bounded and, being integer, it is finite.  Statement (i) then follows from Proposition \ref{prop}.
\end{proof}

\subsection{Interpretation of the results}
Theorem \ref{th_DES} has an interesting biochemical interpretation, where tokens can be seen as molecules. The 
transition operators $T_k$ defined in the proof of Proposition \ref{prop}
remove the tokens (if any) from some nodes (associated with species) so as to generate tokens at other nodes.
The existence of a PLF is equivalent to the fact that no (infinite) sequence of these transitions can drive the token count to infinity at some node.

Consider the important special case of mono-molecular reaction networks, 
where all the internal reactions have the form $X_i\react{g_k} X_j$.
Any transition just moves all tokens in a node (possibly a negative number) to another
node leaving unchanged the total amount. For instance,
if we initialise the network with just a token at node $1$, the set of all possibly reached
states corresponds to a single token at some node.
All mono-molecular reaction networks are associated with bounded GPNs: indeed
they are nonlinear compartmental systems, well known to be structurally stable \cite{MaedaKodamaOhta1978}.

Besides reactions of the form $X_i \react{g_k} X_j$, let us consider internal reactions of the form $X_h+X_r \react{g_s}\emptyset$ and $X_w \react{g_u} \emptyset$. The former new reaction introduces operators such that, if $x_h$ tokens are present at node $X_h$, they are removed ($x_h^+=0$) and the opposite amount appears at node $X_r$, $x_r^+=x_r-x_h$. The latter new reaction removes all the tokens present at node $X_w$ ($x_w^+=0$). Although the total amount of tokens now is not conserved, it cannot increase: hence, these GPNs are also bounded and the corresponding networks are structurally stable, in agreement with \cite{BlanchiniGiordano2014}.

As a simple unbounded case, consider the reactions
$\emptyset  \react{u_1}  X_1$, $X_1\react{g_1} X_2 + X_3$, $X_2 \react{g_2} X_1$,
$X_3 \react{g_3} X_1$, $ X_2 \react{\tilde g_2} \emptyset$, $X_3 \react{\tilde g_3} \emptyset$,
associated with
\begin{equation*}
\begin{cases}
\dot x_1 = u_1 -g_1(x_1) + g_{2}(x_2) + g_{3}(x_3) \\
\dot x_2 =      g_1(x_1)-g_{2}(x_2) -\tilde g_{2}(x_2)\\
\dot x_3 = g_1(x_1)-g_{3}(x_3)-\tilde g_{3}(x_3)
\end{cases}
\end{equation*}
The corresponding Petri net is not bounded. Start with just one token
at $X_1$. Then $\partial g_1/\partial x_1$ can act  producing two tokens,
one at $X_2$ and one at $X_3$. Then $g_2$ and $g_3$ can both act
to transfer the two tokens back at $X_1$.
Repeating the argument, we see an unbounded increase of tokens at node $1$.
Indeed, the Jacobian $J=BDC$, where
\begin{equation*}
B = \left[ \begin{smallmatrix}
-1 & ~1 & ~1 & ~0~  & ~0  \\
~1 & -1 & ~0 & -1~  & ~0 \\
~1 & ~0 & -1 & ~0~  & -1
\end{smallmatrix} \right] ~\mbox{and}~  
C =\left[ \begin{smallmatrix}
1 & 0 & 0\\
0 & 1 & 0\\
0 & 0 & 1\\
0 & 1 & 0\\
0 & 0 & 1
\end{smallmatrix} \right],
\end{equation*}
is not structurally Hurwitz: $\det[-J]=D_1(D_4D_5-D_2D_3)$, the constant term of the characteristic polynomial, can be negative.

\begin{remark}\label{stoi}{\bf (Stoichiometric compatibility class.)}
If the network evolves in a proper stoichiometric compatibility class, namely
$\sigma^\top S=0$ for some vector $\sigma \neq 0$, so that $ \sigma^\top B=0$, we have a conservation law: $\sigma^\top z(t)$ is constant.
This property is preserved by the discrete operators $\Phi_h$:
$\sigma^\top y^+ = \sigma^\top[I+B_h C^\top_h]y = \sigma^\top y$.
For our analysis, we can reduce the system
by applying a state transformation that turns $BDC$ into $T^{-1}BDCT \doteq \hat B D \hat C$ and then neglecting some of the variables (see Example \ref{translation}); the key condition $C_h^\top B_h=  \hat C_h^\top \hat B_h= -1$ is invariant.
Therefore, if $T$ and its inverse are integer matrices, as usually happens,
the proposed theory applies without changes.
\end{remark}

Interestingly, the convergence of Procedure \ref{proc} (i.e., the fact that the procedure stops in finite time) implies that the joint spectral 
radius \cite{Jun09} of the matrix family $\mathcal{F}$ is equal to one:
$$
 \sigma(\mathcal{F})  
\doteq \lim_{k \to 0}\max_{\Phi(\cdot) \in \mathcal{F}} \|  \Phi(1)  \Phi(2)  \Phi(3) \dots  \Phi(k)\|^{\frac{1}{k}} = 1.
$$
Indeed, the convergence of Procedure \ref{proc} implies the boundedness
of the trajectories, hence $\sigma(\mathcal{F})\leq 1$.
Conversely, since $C_k^\top B_k =-1$, matrix $I+B_kC_k$ admits $1$ as an eigenvalue, hence the spectral radius cannot be smaller than $1$.

Boundedness of the GPN evolution is equivalent to polyhedral stability, and the GPN evolves according to an asynchronous mechanism: hence, the different time scales of the system components play no role in defining its stability properties. To formalise this concept, we modify system \eqref{sys} as
 \begin{equation}
\label{time_const}
\Theta \dot x = S g(x) + g_0,
 \end{equation}
where $\Theta=\mbox{diag}\{\theta_1,\theta_2, \dots, \theta_n\} \succ 0$ is a diagonal matrix of positive time constants.
\begin{proposition}
Structural polyhedral stability, guaranteed when Procedure \ref{proc} stops in finite time, implies the structural stability of system \eqref{time_const} for any arbitrary diagonal $\Theta \succ 0$.
\end{proposition}
\begin{proof}
Scaling the state variable as $y = \Theta x$ turns equation \eqref{time_const} into $\dot y  = S g(\Theta^{-1} y) + g_0$. The proof follows immediately by noticing that matrices $B$ and $C$ are the same regardless of $\Theta$, while the derivatives are scaled as $\partial g_i/\partial y_j = \Theta_j^{-1} \partial g_i/\partial x_j$, which does not alter their sign.
\end{proof}

\subsection{Stopping criteria for Procedure \ref{proc}}\label{subsec:stop}
The procedure may fail to converge; in this case, the system does not admit any structural PLF. A possible stopping criterion, proposed in \cite{BlanchiniGiordano2014}, is to interrupt the procedure when either the size of the region $\mathcal{Y}^{k}$ reaches a bound $\mu$, or an assigned maximum number of steps, $n_{max}$, is reached.

We discuss here other possible criteria. The procedure will never converge
if, for some $k$, $\mathcal{Y}^{k}$ includes the original region $\mathcal{Y}^{0}$ in its interior
\cite{BlaMia96,BlaMia15}.  The inclusion of the initial polytope, with vertex matrix $[-I~I]$, should be checked at each step; this can be done as follows.
\begin{proposition}
If the polytope  $\mathcal{Y}=\mbox{conv}\{[-Y,Y]\}$ has a non-empty interior, then it includes $\mathcal{Y}^{0}=\mbox{conv} \{[-I,I]\}$
in its interior if and only if 
 \begin{equation} \label{lpe}
\nu_i = \min \{\| p_i \|_1:  Yp_i  = e_i \} < 1,~~~i=1,2,\dots,n,
  \end{equation}
where $\|p\|_1 = \sum |p_i|$ is the $1$-norm and $e_i$ is the $i$th vector of the canonical basis.
\end{proposition}
\begin{proof}
From expression \eqref{polyfuncX} we see that 
$\nu_i = V_Y(e_i)$, where $V_Y$ is the polyhedral norm  with unit ball
$\mathcal{Y}$. Vector $e_i$ (along with its opposite $-e_i$) is in the interior of  $\mathcal{Y}$
iff $\nu_i=V_Y(e_i) < 1$. Moreover, $\mathcal{Y}^0= \mbox{conv} \{[-I,I]\}$
is in the interior of  $\mathcal{Y}$ iff it vertices, i.e. $\pm e_i$,
are in the interior.
\end{proof}

Note that problem \eqref{lpe} can be solved via linear programming and allows to efficiently stop the procedure at an early stage when the system does not admit a structural PLF.

Another stopping criterion relies on eventually periodic matrices.
A square matrix $M$ is said \emph{eventually periodic} if there 
exist a non-negative integer $m$ and a positive integer $p$ such that 
\begin{equation}
M^m = M^{m + kp} \qquad ~~\mbox{for all integer}~ k \geq 0.
\end{equation}
\begin{proposition}\label{roots}
If the square matrix $M$ is eventually periodic, then its eigenvalues are either zero or roots of the unity.
\end{proposition}
\begin{proof}
Take the eigenpair $(\lambda,v)$, $Mv = \lambda v$, $v\neq 0$. If $M$ is eventually periodic, $$(M^m - M^{m + kp})v =  \lambda^m (1 - \lambda^{k p}) v = 0.$$
Then $\lambda^m (1 - \lambda^{k p})=0$, i.e., $\lambda$ must be either $0$ 
or a root of the unity.
\end{proof}
\begin{proposition} \label{evper}
The set $\Pi(\mathcal{F})$ is a finite set only if each matrix in $\Pi(\mathcal{F})$ is eventually periodic.
\end{proposition}
\begin{proof}
By contradiction, if a matrix in $\Pi(\mathcal{F})$ is not eventually periodic, then its powers form an infinite sequence of different matrices, hence the set $\Pi(\mathcal{F})$ is infinite.
\end{proof}

The previous condition is not sufficient: even if all matrices are eventually periodic,  there can be an infinite sequence of products among them. Combining Propositions \ref{roots}  and  \ref{evper} yields the following corollary.
\begin{corollary}
\label{Cor:Corollary}
If $\Pi(\mathcal{F})$ includes a matrix whose eigenvalues are not either zero or roots of the unity, then it has infinite cardinality. 
\end{corollary}

\begin{remark}
To check the condition, there is no need to compute \emph{all} the products in $\Pi(\mathcal{F})$. In fact,
\begin{itemize}
\item if $(I+B_iC_i^\top) \in \mathcal{F}$, then $(I+B_iC_i^\top)^n = (I+B_iC_i^\top)$ for any positive integer $n$. 
Indeed, since $C_i^\top B_i=-1$,  
$$
(I+B_iC_i^\top)^2 = I+ 2B_iC_i^\top + B_i(C_i^\top B_i)C_i^\top =I+B_iC_i^\top;
$$
\item if $\Phi_1=I+B_1C_1^\top$ and $\Phi_2=I+B_2C_2^\top$ with $C_1 = C_2$, then $\Phi_1\Phi_2=\Phi_2$,
as it can be seen in a similar way.
\end{itemize}
\end{remark}

Based on Corollary \ref{Cor:Corollary}, an alternative stopping criterion is achieved by computing the sequence of products  $\Pi(\mathcal{F})$ of increasing order
and stopping whenever one of them has an eigenvalue that is neither zero nor a root of the unity. This produces, in principle, an exponentially growing list of matrices. Yet,
extensive numerical experiments have shown that, in most cases, the stopping condition is quickly reached.

\subsection{Turning nodes into black holes}\label{subsec:blackholes}
We introduce a new type of node, called \emph{black hole},
in which any incoming token (either positive or negative) is cleared, 
so that the black hole contains zero tokens throughout the system evolution.
If a node is replaced by a black hole, then the GPN associated with the reaction network is transformed and behaves differently.

\begin{example}\label{ex:Vivaldi}
Consider the reaction network in Fig. \ref{fig:V}, where $\emptyset \react{} X_1$, $\emptyset \react{} X_2$, $X_1+X_2 \react{} X_3+X_4$, $X_4 \react{} X_2$, $X_1+X_3 \react{} \emptyset$. The GPN associated with this network is unbounded, i.e., no structural PLF exists for the system.
If we turn $X_2$ into a black hole, then we virtually have $X_2 := \emptyset$ and the transformed reaction network becomes: $\emptyset \react{} X_1$, $X_1  \react{} X_3+X_4$, $X_4 \react{} \emptyset$, $X_1+X_3 \react{} \emptyset$.
The GPN associated with the transformed network is bounded.
Boundedness is achieved also if node $X_4$ is turned into a black hole, instead of $X_2$. Conversely, turning either node $X_1$ or node $X_3$ into a black hole does not yield boundedness.
\end{example}

\begin{figure}[tb]
\centering
\includegraphics[width=5cm]{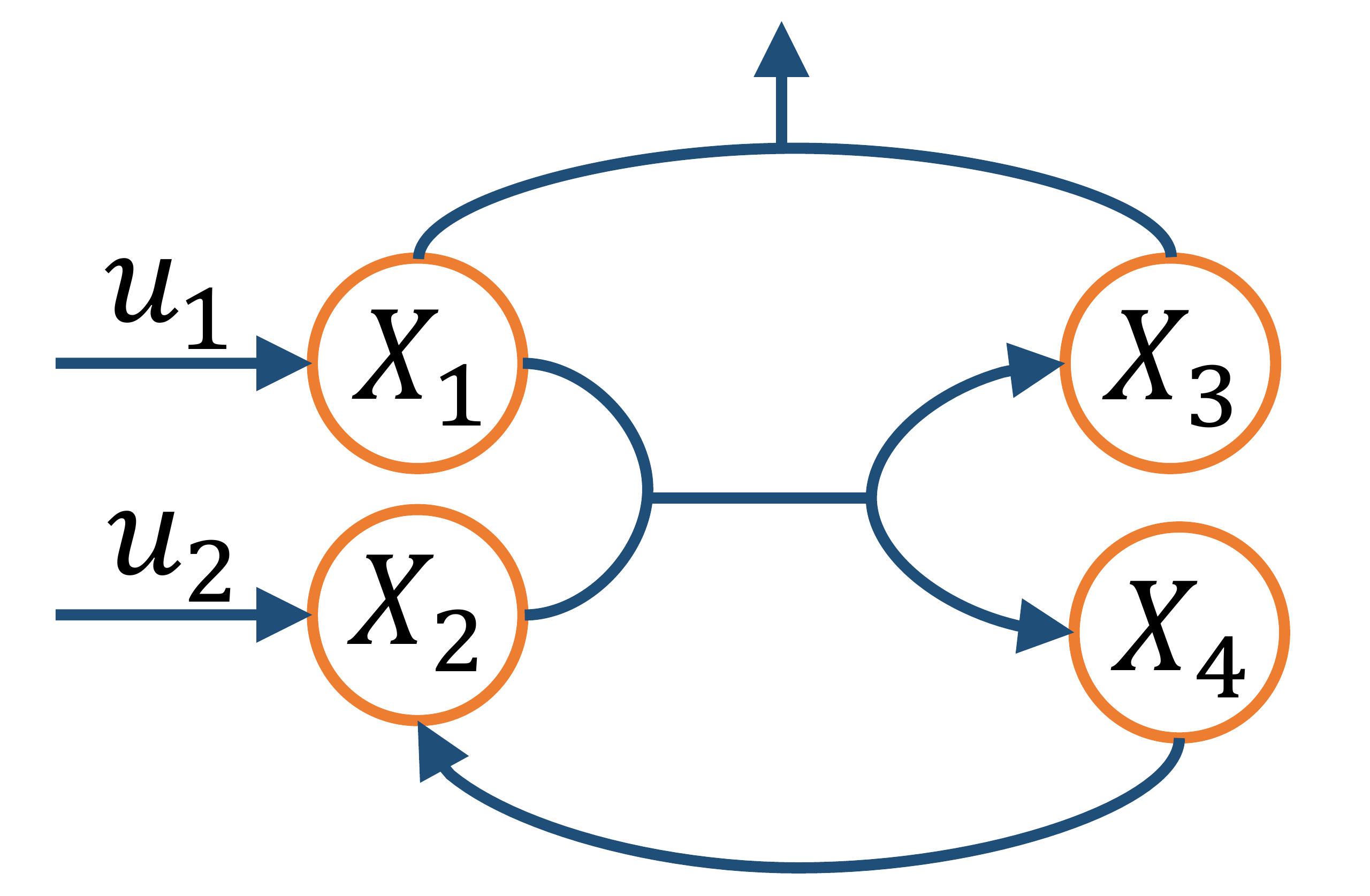}
\caption{Graph of the chemical reaction network corresponding to the system in Example \ref{ex:Vivaldi}.}
\label{fig:V}
\end{figure}

As we show in the next section, replacing a node with a black hole can be regarded as applying a strong feedback to that node, thus enforcing a pinning control.

\section{Pinning Control}\label{sec:pinningcontrol}

Pinning control strategies control just some of the state variables, each by means of a (strong) local feedback, so as to stabilise the whole network.
Without loss of generality, we split the state vector as $z = [z_1^\top \,\, z_2^\top]^\top$ and we assume that a feedback control action with $\gamma>0$ is applied to $z_1$, the first $p$ entries of $z$:
$$
\dot z_1 = S_1(g(x)-g(\bar x)) - \gamma z_1,
$$
where $S_1$ contains the first $p$ rows of $S$. The system \eqref{bdc} can then be split as
\begin{equation}\label{pinned}
\begin{bmatrix}\dot z_1  \\\dot z_2  
\end{bmatrix}
=  \underbrace{\begin{bmatrix}
B_1D(z )C_1 -\gamma I& B_1D(z )C_2\\B_2D(z )C_1 & B_2D(z )C_2
\end{bmatrix}}_{A(D(z),\gamma)}
\begin{bmatrix}
z_1   \\z_2  
\end{bmatrix},
\end{equation}
 where  $B_1$ and $B_2$ contain  the first $p$ and the last $n-p$ rows of $B$,
respectively, while $C_1$ and $C_2$ contain the first $p$ and the last $n-p$ columns of $C$. 

We have a first main result.
\begin{theorem}\label{theorem:pinn}
Under Assumptions \ref{as_eq}-\ref{largebound}, the following conditions are equivalent:
\begin{itemize}
\item[(i)]
The $(n-p)$-dimensional subsystem 
\begin{equation}\label{non-pinned}
\dot z_2 (t)=  B_2 D(t) C_2 z_2(t)
\end{equation}
admits a structural PLF $U(z_2)$, in the strong sense of Definition \ref{def:Lyap}.

\item [(ii)] There exists a polyhedral function $V(z)$ and $\bar \gamma >0$ such that $V(z)$ is a structural strong LF
for  system \eqref{pinned} if $\gamma \geq \bar \gamma$ .
\end{itemize} 
\end{theorem}

\begin{remark}
The structural result needs Assumption \ref{largebound}. 
Consider the linear differential inclusion with matrix
$$\begin{bmatrix}
-\gamma+ D_1(t) & 1\\
1 & -D_2(t)
\end{bmatrix}.$$
No matter how large $\gamma>0$ is, the system becomes unstable 
if $D_1(t)>0$ diverges and/or $D_2(t)>0$ converges to $0$;
conversely, under Assumption \ref{largebound}, a stabilising $\gamma$ always exists, which depends on the bounds in \eqref{largeboundeq}.
\end{remark}

{\bf Proof of Theorem \ref{theorem:pinn}}

{\it (i) $\Rightarrow$ (ii).} Due to the compactness Assumption  \ref{largebound},
the matrix family in system \eqref{pinned} is polytopic and we can write
$$
A(D,\gamma)=\begin{bmatrix}
B_1DC_1 & B_1DC_2\\B_2DC_1 & B_2DC_2
\end{bmatrix}
=
\sum_{k=1}^M\lambda_{ k} 
\begin{bmatrix}
A^{(k)}_{11}  &  A^{(k)}_{12}\\
A^{(k)}_{21}  &  A^{(k)}_{22} 
\end{bmatrix},$$
with $\sum_{k=1}^M\lambda_{k}=1$, $\lambda_{k}>0$, because each of the four matrix blocks in $A(D,\gamma)$ can be expressed as a polytopic matrix 
$\sum_{k=1}^M\lambda_{k} A^{(k)}_{ij}$ \cite{BlaMia15}.

By assumption, system \eqref{non-pinned} admits a structural PLF, hence $[A^{(k)}_{22}] X_{ 2 } = X_{ 2 } P_{22}^{(k)}$,
where $P_{22}^{(k)}$ is strictly column diagonally dominant and $X_2$ has full row rank
\cite{BlaMia15,BraTon80,Molchanov_Pyatnitskiy_86}. 
Then we can take
$$
\hat X = 
\begin{bmatrix}
\rho I & 0\\
0 & X_2
\end{bmatrix},
$$
where $\rho>0$ is a parameter to be selected, and write the $\gamma$-parametrised equation
\begin{eqnarray}
\underbrace{\left [\begin{smallmatrix}
-\gamma I +A^{(k)}_{11} & A^{(k)}_{12}\\
A^{(k)}_{21} &  A^{(k)}_{22}
\end{smallmatrix}\right ]}_{\hat A^{(k)}}
\hat X = \hat X
\underbrace{
\left [\begin{smallmatrix}
-\gamma I   + A^{(k)}_{11} & A_{12}^{(k)}X_2/\rho\\
M_{21}^{(k)}\rho & P_{22}^{(k)}
\end{smallmatrix}\right ]}_{\hat P^{(k)}} \label{bigeq}
\end{eqnarray}
where $M_{21}^{(k)}$ are matrices that satisfy $X_2 M_{21}^{(k)} =A_{21}^{(k)}$, and always exist because $X_2$ has full row rank.

Since $ P_{22}^{(k)}$ is strictly diagonally dominant, we can take $\rho>0$
large enough, so that the diagonal dominance is preserved in the last $n-p$ columns of $\hat P^{(k)}$.
Now, for any choice of $\rho$, there exists a $\bar \gamma$ large enough 
such that, for  $\gamma \geq \bar \gamma$,  diagonal dominance of the first $p$ columns of $\hat P^{(k)}$
is ensured. The resulting equations $\hat A^{(k)} \hat X =\hat X \hat P^{(k)}$,
with $P^{(k)}$ diagonally dominant, ensure that $[-\hat X,\hat X]$ are the vertices of a PLF for system \eqref{pinned}.

{\it  (ii) $\Rightarrow$ (i).} A technical Lemma is required.
\begin{lemma} \label{sublemma}
Assume that the convex and compact set ${\cal S}$ including the origin
as an interior point (C-set) is positively invariant for the linear time-invariant system
\begin{equation}\label{systemmu}
\begin{bmatrix}
\dot z_1 \\
\dot z_2
\end{bmatrix} = \begin{bmatrix}
F_{11}-\mu I & F_{12}\\
F_{21} & F_{22}
\end{bmatrix}
\begin{bmatrix}
 z_1 \\
 z_2
\end{bmatrix}
\end{equation}
for all $\mu >0$. 
Then the intersection ${\cal S}_2 =  \{z\in {\cal S} \colon z_1=0\}$ is a C-set  in the subspace 
$[z_1^\top ~z_2^\top]^\top$ with
$z_1=0$ and is positively invariant for the subsystem $\dot z_2 = F_{22}z_2$.
\end{lemma}

The proof of  Lemma \ref{sublemma} is in the appendix.

By assumption, ${\cal S} = \{z \colon V(z) \leq 1\}$, the unit ball of $V(z)$,
is an invariant set for system \eqref{pinned}
for all $\gamma \geq \bar \gamma$. 
Let us perturb system \eqref{pinned} and write it as
\begin{equation}\label{pinnedeps}
\dot z 
=  [A(D,\gamma) + \epsilon I]z \doteq F_\epsilon(D,\gamma) z,
\end{equation}
with $\epsilon>0$ small enough to ensure that exponential stability is preserved:
if $D^+V(z) \leq -\beta V(z)$,  just take $0<\epsilon<\beta$.

For any fixed $\bar D$, if $\gamma \geq \bar \gamma$, ${\cal S}$ is invariant
for the linear time invariant system
with state matrix $F_\epsilon(\bar D,\gamma)$. 
Take $\mu =\gamma -\bar \gamma$ 
and apply Lemma \ref{sublemma}. The intersection ${\cal S}_2  = \{z \in {\cal S} \colon z_1=0\}$,
which is a polyhedral C-set in the $z_2$-space, 
is positively invariant for the system
$\dot z_2 = [A_{22}(\bar D,\gamma)+\epsilon I] z_2$. 
Since this claim is true for any choice of $D$, the C-set  ${\cal S}_2$
is robustly positively invariant for the differential inclusion
$\dot z_2 = [A_{22}(D(t),\gamma)+\epsilon I]z_2$, which is thus at least marginally stable. As a consequence,
$\dot z_2 = A_{22}(D,\gamma)z_2$ is exponentially stable, hence it admits a PLF
\cite{BlaMia15,BraTon80,Molchanov_Pyatnitskiy_86}.
\QED

Although the same local feedback parameter $\gamma$ is considered in Theorem \ref{theorem:pinn} for all pinned nodes, different parameters $\gamma_i$ could be adopted for the $p$ nodes, provided that, for all $i$, $\gamma_i \geq \bar \gamma$ is large enough to ensure diagonal dominance of the  first  $p$ columns in the last matrix in \eqref{bigeq}. Furthermore, the result easily extends to nonlinear feedback strategies $k_i(z_i) z_i$, provided that $k_i(z_i) \geq \bar \gamma$ for all $z_i$. 

We now need to face a technical issue.
Procedure \ref{proc} can be adopted to find a structural PLF for 
system \eqref{non-pinned}, and all the results in Section \ref{polybounded},
including the stopping criterion, remain valid.
Unfortunately, the procedure provides a \emph{weak} structural PLF $V(z_2)$, and not a \emph{strong} one as required by Theorem \ref{theorem:pinn}.
 To fix the problem we consider three facts.
\begin{itemize}
\item In view of Theorem \ref{nonsing}, if we find a weak PLF for the differential inclusion \eqref{non-pinned},  we can claim its robust asymptotic stability if (and only if) matrix $BDC$ is structurally non-singular.
\item A classical result \cite{BlaMia15,BraTon80,Molchanov_Pyatnitskiy_86}
ensures that, if  \eqref{non-pinned} with compact bounds \eqref{largeboundeq} (introduced exactly for this technical reason) is asymptotically stable, then
it is also exponentially stable and it admits a strong PLF $U(z_2)$.
\item Hence, if we find $V(z_2)$ and $BDC$ is structurally non-singular, we 
know that a strong PLF $U(z_2)$ exists. This allows us to apply Theorem \ref{theorem:pinn}; fortunately, we do not need to compute $U(z_2)$.
\end{itemize} 
\begin{corollary} \label{corro}
Assume that $V(z_2)$ is a weak structural PLF for
\eqref{non-pinned} and that $\det [B_2DC_2]\neq 0$ for all $D$ in \eqref{largeboundeq}.
Then, there exists $\bar \gamma$ such that, for all $\gamma> \bar \gamma$,
system \eqref{pinned} is structurally exponentially stable.
\end{corollary}
Structural non-singularity of $BDC$ can be checked as discussed in Remark \ref{rem:nons}.

\begin{remark}
All the results presented in this section so far hold also for non-unitary networks: we do not need to assume $C_k^\top B_k=-1$.
\end{remark}

\subsection{Lyapunov function for the free variables $z_2$}
If we find a PLF for the $z_2$-subsystem of dimension $n-p$, and we have structural non-singularity of $BDC$, then the stability of the overall system is ensured for large enough $\gamma$, according to Corollary \ref{corro}. How can we exploit the GPN to this aim?
\begin{theorem}\label{theorem:FSM}
If $C_i^\top B_i = -1$ $\forall i=1,\dots,q$, the following statements are equivalent:
\begin{itemize}
\item [i)] The $z_2$-subsystem \eqref{non-pinned} admits a weak structural PLF.
\item  [ii)] The evolution of the GPN where nodes $1, \dots, p$ have been turned into black holes is bounded for any integer initial marking.
\end{itemize}
\end{theorem}
\begin{proof}
We show that a (weak) PLF exists for \eqref{non-pinned} if and only if
a (weak) PLF exists for the differential inclusion  
\begin{equation}\label{larga}
\begin{bmatrix}
\dot z_1 \\
\dot z_2
\end{bmatrix}
=\begin{bmatrix}
0 & 0\\
0 &  B_2D(t)C_2
\end{bmatrix}
\begin{bmatrix}
z_1 \\
 z_2
\end{bmatrix}.
\end{equation}
Indeed, let $B_2DC_2 =\sum_{k=1}^q~D_k B_{2,k} C_{2,k}^\top$.
A weak PLF for \eqref{larga} exists if and only if the equation \cite{BlaMia15}
 \begin{equation} 
\label{full}
\begin{bmatrix}
0 & 0\\
0 & A_{22}^{(k)}
\end{bmatrix}
\begin{bmatrix}
X_1 \\
X_2
\end{bmatrix}
=
\begin{bmatrix}
X_1 \\
X_2
\end{bmatrix}
P ^{(k)}
 \end{equation} 
holds with $X=[X_1^\top X_2^\top]^\top$ full row rank and $P ^{(k)}$ weakly diagonally dominant. Hence, $X_2$ has full row rank and
\begin{equation} A_{22}^{(k)}X_2 = X_2P_{22}^{(k)} \label{reduced}
\end{equation} 
holds, which is equivalent to the existence of a weak PLF for the subsystem \eqref{non-pinned}.
Conversely, if \eqref{reduced} holds with $X_2$ full row rank and $P_{22}^{(k)}$ diagonally dominant,
then \eqref{full} holds with $X$ and $P^{(k)}$ as follows
$$
\begin{bmatrix}
0 & 0\\
0 & A_{22}^{(k)}
\end{bmatrix}
\begin{bmatrix}
I & 0 \\
0 & X_2
\end{bmatrix}
=\begin{bmatrix}
I & 0 \\
0 & X_2
\end{bmatrix}
\begin{bmatrix}
0 & 0\\
0 & P_{22}^{(k)}
\end{bmatrix}.
$$
To complete the proof, note that, after having turned nodes $1, \dots, p$ into black holes, the evolution of the GPN is represented by the integer operators
$$
\Psi(k)=
\begin{bmatrix}
0 & 0\\
0 &  I+B_{2,h} C_{2,h}^\top
\end{bmatrix}
$$
If we apply Procedure \ref{proc}, initialised with $X_0=[I,-I]$,
we generate the matrices $\mathbb{X}_{k}$ 
$$\mathbb{X}^{0}=
\left [
\begin{array}{rr|rr}
I & 0 & -I & 0\\
0 &  I& 0 & -I
\end{array}
\right ]~\mbox{and}~
\mathbb{X}_{k}=
\begin{bmatrix}
0\\
 \mathbb{X}_2^{k}
\end{bmatrix},~~k>0
$$
where the first $p$ rows are zero, while the remaining $n-p$ rows are exactly those we get by applying the procedure to subsystem \eqref{non-pinned}.
On the other hand, Procedure \ref{proc} applied to the original system with the pinned nodes converges if and only if a (weak) structural PLF exists for \eqref{non-pinned}, as shown in \cite{BlanchiniGiordano2014}.
\end{proof}
 
The conditions of Theorem \ref{theorem:FSM} are equivalent to the fact that Procedure \ref{proc} converges after having zeroed the rows of $B$ and the columns of $C$ corresponding to the pinned nodes.
From a computational standpoint, however, applying the procedure on the 
subsystem \eqref{non-pinned} is more convenient, because the system size is reduced and the stopping criteria discussed in Section \ref{subsec:stop} remain valid for the restricted subspace, while they are no longer valid for the original state space.

\begin{remark} \emph{(Strong convergence.)}  According to Theorem \ref{theorem:FSM}, the differential inclusion can be reduced as in  \eqref{larga}. If matrix $B_2DC_2$ is structurally non-singular, then pinning the first $p$ nodes makes the differential inclusion \emph{strongly convergent}, according to the definition in \cite{FBBCG20}.
\end{remark} 

It is worth stressing that a \emph{dual procedure} can be adopted as well.
We can derive a Lyapunov function (if it exists) defined in terms of planes, as in \eqref{polyfuncF}. As shown in \cite{BlanchiniGiordano2014,BlanchiniGiordanoCDC2015}, this is equivalent to
applying the Procedure \ref{proc} to the dual system
$\dot w(t) = C^\top D(t) B^\top w(t)$.
In view of duality properties \cite{BG2021,BlaMia15,Molchanov_Pyatnitskiy_86}, a PLF exists for the primal system if and only if it exists for its dual.

\subsection{Arc pinning: Regulating the reactions}
Node pinning means imposing a strong feedback to some nodes. By arc pinning, we mean that a strong feedback is imposed to some flows, hence
$$g_j(x) := g_j(x) - \gamma(g_j(x) - \bar r_j),
$$
so that, roughly speaking, the flow is forced to have a prescribed nominal value.

This \emph{dual} arc-pinning problem can be solved by writing the system in reaction coordinates, as in \cite{AlrAng13,AlrAng16,AlrAngSon19}, and then adopting the $EDF$-decomposition \cite{BlanchiniGiordanoCDC2015}.
If $\bar{x}$ is an equilibrium point of the system, under suitable conditions, we can define a transformation from concentration coordinates $(x_i)_{i=1}^\numconc$ to reaction coordinates $(r_j)_{j=1}^\numreac$ as $r(x(t)) = g(x(t)) - g(\bar{x})$. After this transformation, the system becomes
 
\begin{equation}\label{Eq:EDFpreDecomposition}
\dot{r}(t) = \bigg[ \frac{\partial g}{\partial x} \bigg]Sr(t).
\end{equation}
A procedure similar to the $BDC$-decomposition \cite{BlanchiniGiordanoCDC2015} transforms system \eqref{Eq:EDFpreDecomposition} into the linear differential inclusion
\begin{equation}\label{Eq:EDF_decom}
\dot{r}(t) = ED(t)Fr(t),
\end{equation}
where $D$ is a diagonal matrix with positive diagonal entries.
The theory remains completely unchanged. Pinning reactions $1,\dots ,p$ is equivalent to zeroing the first $p$ rows of $E$ and the first $p$ columns of $F$.

\subsection{Periodic forcing input}
Assume that $g_0(t)$ is a periodic input \cite{CourSepul2009,FioGuaDiB19,russo_dibernardo_sontag09} and $x_p(t)$
is a periodic target trajectory, having the same period as $g_0(t)$,
corresponding to  $g_0(t)$. Let $x(t)$ be any other trajectory.
Then, we can write
$$
\dot x(t) = Sg(x(t)) + g_0(t) \,\, \mbox{and} \,\, \dot x_p(t) = Sg(x_p(t)) + g_0(t).
$$
Denoting $z(t)=x(t)-x_p(t)$, the $BDC$-decomposition leads to the dynamical system
$$
\dot z(t) = B D(z(t),t) C z(t).
$$
Since our analysis considers a differential inclusion with \emph{arbitrary} time-varying $D(\cdot)$, the stability -- or stabilisation via pinning control -- of such a differential inclusion implies $z(t) \rightarrow 0$, hence the stability -- or stabilisation -- of the periodic trajectory.

\section{Pinning Control of Reaction Networks}\label{sec:applications}

Pinning control, as mentioned, consists in applying strong local feedback actions to some
nodes (or arcs) with the aim of regulating the whole network. How can one select the node(s) to be pinned in order to achieve the control objective?

Based on our results, we can re-formulate the question as: Which are the nodes that, if converted into black holes, ensure global boundedness of the GPN evolution, hence leading to Lyapunov stability of the overall system once they are subject to a sufficiently strong local feedback?  

We provide here some examples of chemical reaction networks: the readers are invited to have a preliminary look at the network graphs (shown in Figures \ref{fig:R} and \ref{fig:M}) and see if they can spot immediately which nodes are the \emph{most important} ones to be governed so as to rule all the others; the authors of this paper often failed to guess these nodes in advance.

\begin{figure}[tb]
\centering
\includegraphics[width=7cm]{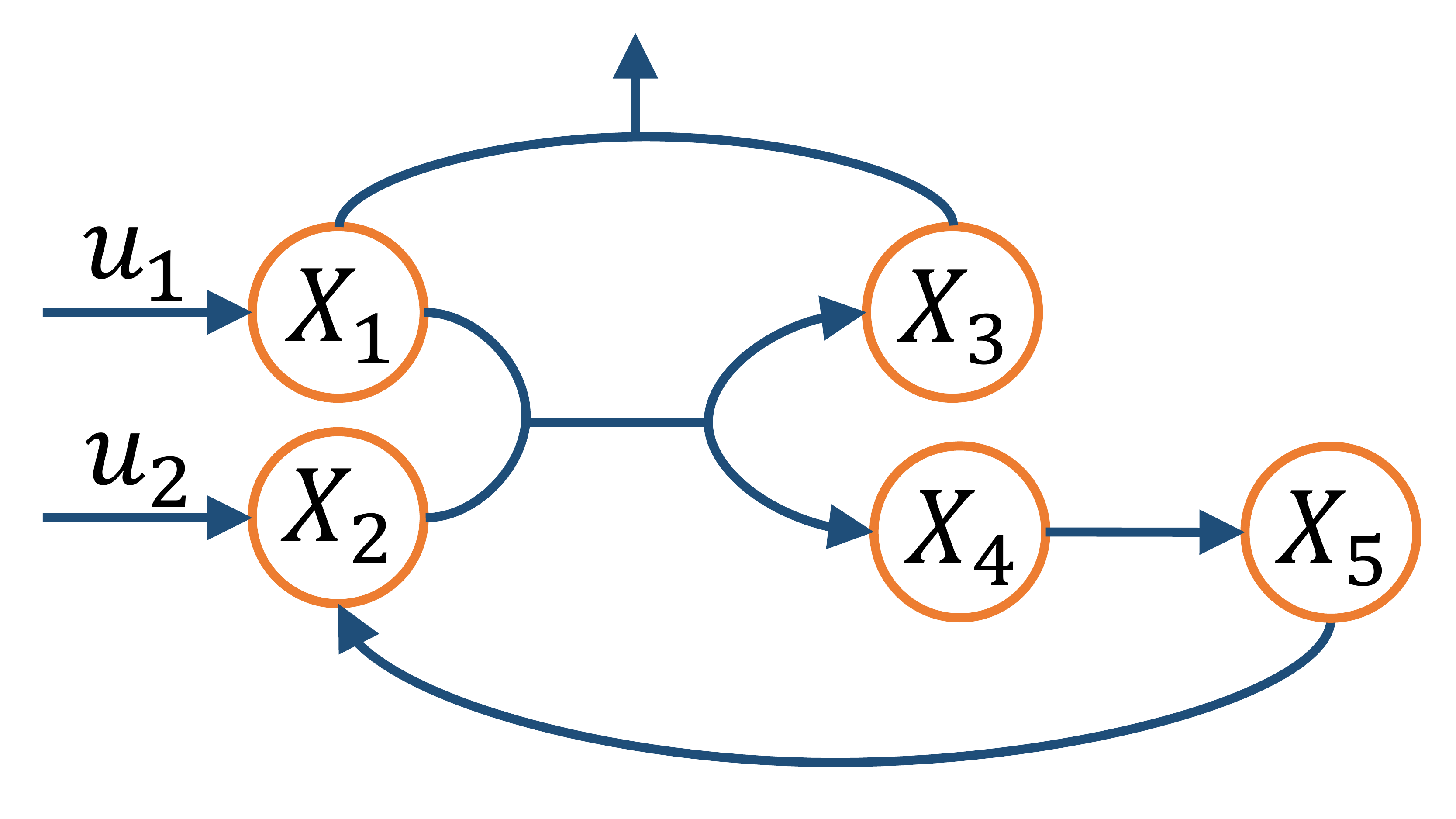}
\caption{Graph of the chemical reaction network corresponding to the system in Example \ref{ex:Ravel}.}
\label{fig:R}
\end{figure}

\begin{example}\label{ex:Ravel}
For the network in Fig. \ref{fig:R}, Procedure \ref{proc} does not converge, hence the system does not admit a structural polyhedral Lyapunov function. However, if we pin any of the nodes $X_2$, $X_4$ or $X_5$, the procedure converges, hence the system is structurally stabilised by enforcing a sufficiently strong local feedback on any of these nodes.
Conversely, the procedure does not converge even if we pin node $X_1$ or node $X_3$. This fact is explained by noticing that pinning $X_2$, $X_4$ or $X_5$ cuts the loop $X_2 \rightarrow X_4 \rightarrow  X_5 \rightarrow X_2$, and tokens repeatedly circulating in this loop continue depleting $X_1$ or filling up $X_3$.
\end{example}

\begin{figure}[tb]
\centering
\includegraphics[width=7cm]{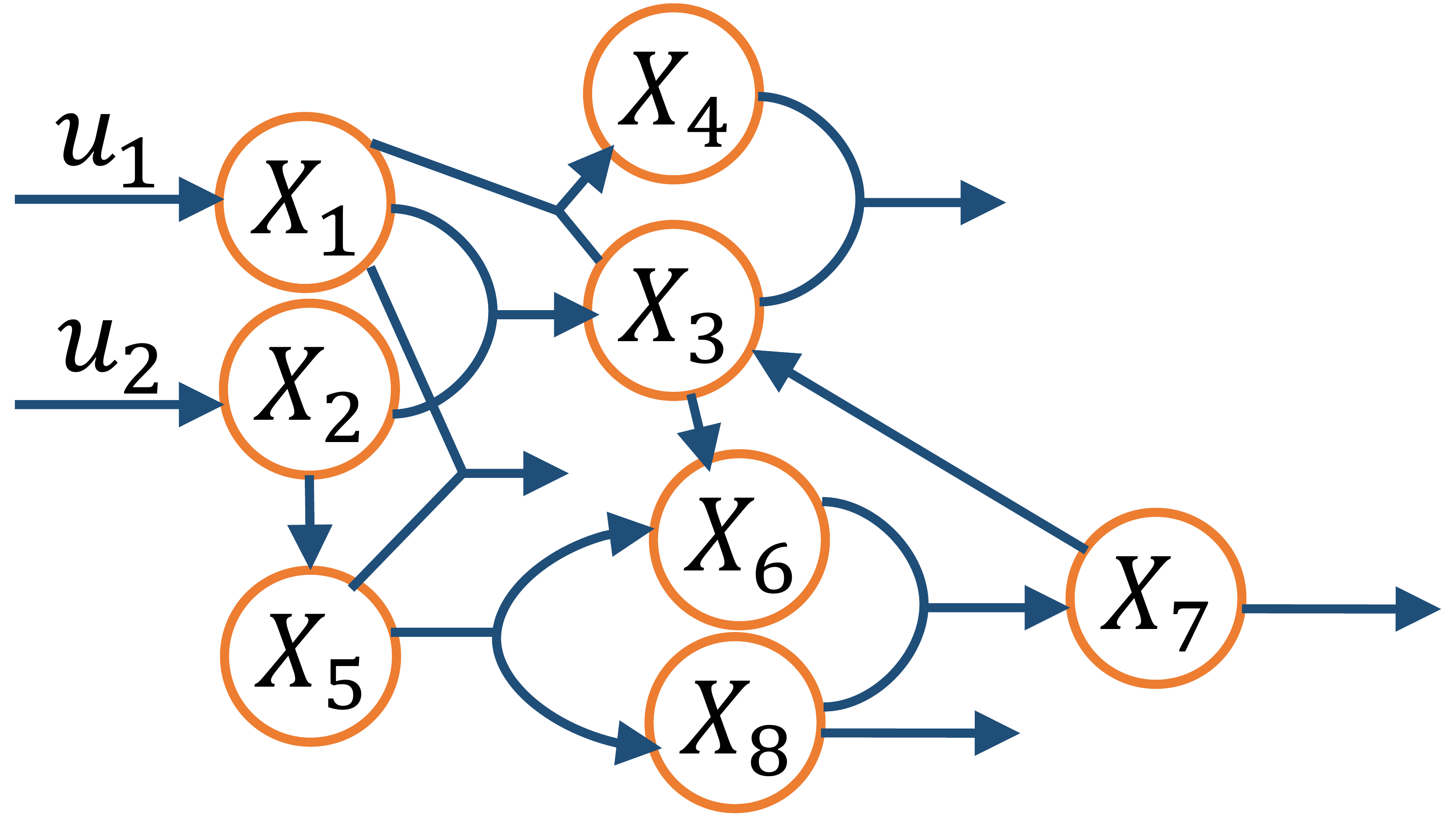}
\caption{Graph of the chemical reaction network corresponding to the system in Example \ref{ex:Massenet}.}
\label{fig:M}
\end{figure}

\begin{example}\label{ex:Massenet}
The network in Fig. \ref{fig:M} does not admit a structural PLF: Procedure \ref{proc} does not converge.
If we pin node $X_3$, then the procedure converges.
This node is then the most crucial one in the network: pinning any one of the other nodes does not yield convergence.
\end{example}

\begin{example}\label{translation}
\textbf{(A translation model.)}
The complete translation model proposed in \cite{DelVecchioMurray2014} includes the chemical reactions:
\begin{equation}\label{react:translation}
\begin{array}{c}
X_1+X_2 \revreact{g_{12}}{g_3} X_3 \react{g^*_3} X_4+X_2,\\
X_4 \react{g_4} X_5 \react{g_5} X_6\react{g_6} X_7\react{g_7} X_8 + X_9,\\
X_8 \react{g_8}X_1, \qquad X_9 \react{g_9}   X_{10} \react{g_{10}} \emptyset.
\end{array}
\end{equation}

The stability of the reduced order model \cite{DelVecchioMurray2014}
$$
X_1+X_2 \revreact{g_{12}}{g_{3}}  X_3 \react{g^*_3} X_1+ X_2+X_4,~~~X_4\react{g_4} \emptyset
$$
is considered in \cite{AlrAngSon19}, where the system is shown to admit a piecewise-linear Lyapunov function in rates.

We consider here the complete model associated with the reaction network \eqref{react:translation}, corresponding to the system of equations
$$
\begin{array} {rl}&
\left \{
\begin{array}{rcl}
\dot x_1 &=& -g_{12}(x_1,x_2) + g_{3}(x_3) +g_{8}(x_8)  \\
\dot x_2 &=& -g_{12}(x_1,x_2) + g_{3}(x_3) +g_{3}^*(x_3)\\
\dot x_3 &=&  g_{12}(x_1,x_2) - g_{3}(x_3) -g_{3}^*(x_3)\\
\dot x_4 &=&  g_{3}^*(x_3) -g_{4}(x_4) \\
\dot x_5 &=&   g_{4}(x_4)-g_{5}(x_5) \\
\dot x_6 &=&   g_{5}(x_5)-g_{6}(x_6) \\
\dot x_7 &=&   g_{6}(x_6)-g_{7}(x_7) \\
\dot x_8 &=&   g_{7}(x_7)-g_{8}(x_8)
\end{array}
\right .
\\&\hspace{2mm}
\begin{array}{rcl}
\dot x_9  &=& g_{7}(x_7)-g_{9}(x_9)\\
\dot x_{10} &=& g_{9}(x_9)- g_{10}(x_{10})
\end{array}
\end{array}
$$
We do not consider the last two equations: if the subsystem associated with
the variables $x_1$-$x_8$ converges to an equilibrium, and in particular $x_7 \rightarrow \bar x_7$, then also $x_9$ and, in turn, $x_{10}$ converge to an equilibrium.

It is apparent that  
  $\dot x_1 +\dot x_3 + \dot x_4 +\dot x_5 +\dot x_6 +\dot x_7 + \dot x_8=0 $ 
and $\dot x_2+\dot x_3=0,$ hence the sums of concentrations
$x_1+x_3+ x_4+x_5+x_6+x_7+x_8 \doteq w_1$
and
$x_2+x_3 \doteq w_2$
remain constant, thus forming a stoichiometric compatibility class, which is bounded because $x \geq 0$. Hence, an equilibrium indeed exists \cite{BG2021survey}.
We can then replace variables $x_1$ and $x_2$ by $w_1$ and $w_2$,
and keep all the others variables. This is equivalent to considering the state transformation $T^{-1}x=w$ and the transformed system
$$
\dot w = T^{-1}BDCT w,
$$
with
$$
 T^{-1}= \left[
\begin{smallmatrix}
     1   &  0  &   1   &  1    & 1    & 1   &  1   &  1\\
     0   &  1    & 1    & 0   &  0   &  0   &  0    & 0\\
     0    & 0   &  1  &   0  &   0 &    0  &   0   &  0\\
     0   &  0  &   0   &  1   &  0   &  0    & 0    & 0\\
     0   &  0  &   0   &  0   &  1   &  0  &   0   &  0\\
     0   &  0  &   0   &  0  &   0  &   1   &  0   &  0\\
     0   &  0  &   0   &  0   &  0   &  0   &  1  &   0\\
     0   &  0  &   0   &  0   &  0   &  0   &  0   &  1
 \end{smallmatrix}
\right]
$$
and then neglecting the first two rows of $B$, which are of course zero
because $w_1$ and $w_2$ are both constant, as well as the first two columns of $C$. The resulting reduced $6$-dimensional system admits the $BDC$-decomposition with
$$B=
\left[
\begin{smallmatrix}
     1   &  1   & -1   &  0   &  0  &   0    & 0  &   0\\
     0  &   0  &   1  &  -1  &   0  &   0  &   0  &   0\\
     0  &   0   &  0   &  1   & -1  &   0  &   0   &  0\\
     0  &   0   &  0  &   0 &    1 &   -1  &   0  &   0\\
     0  &   0   &  0  &   0  &   0 &    1  &  -1  &   0\\
     0  &   0  &   0  &   0  &   0 &    0  &   1  &  -1
\end{smallmatrix}
\right]$$
and
$$C=\left[
\begin{smallmatrix}
    -1  &  -1   & -1   & -1  &  -1  &  -1\\
    -1   &  0  &   0    & 0   &  0  &   0\\
     1   &  0   &  0   &  0  &   0  &   0\\
     0  &   1   &  0   &  0  &   0  &   0\\
     0  &   0  &   1   &  0  &   0   &  0\\
     0  &   0   &  0   &  1  &   0  &   0\\
     0  &   0   &  0   &  0  &   1  &   0\\
     0  &   0   &  0   &  0  &   0  &   1
 \end{smallmatrix}
\right],
$$
where $C_i^\top B_i=-1$, for all $i$, as expected (cf. Remark \ref{stoi}).
Procedure \ref{proc} converges, yielding a structural PLF with 42 vertices, and matrix $BDC$ passes the structural non-singularity test: this proves asymptotic stability of the complete system.
\end{example}
 
\begin{example}
\textbf{(A transcription model.)}
The complete transcription model proposed in \cite{DelVecchioMurray2014} is
\begin{equation}\label{react:transcription}
\begin{array}{c}
X_1 \revreact{g_{1}}{g_2} X_2  \revreact{g^*_{2}}{g_3} X_3,\\
X_3+X_4\revreact{g_{34}}{g_5} X_5 \react{g_5^*}X_6 \react{g_6}X_4+X_7,\\
X_7 \react{g_7} X_8 \react{g_8} X_9 \react{g_9} X_{10} \react{g_{10}} X_{11}+X_{12}, \\
X_{11}\react{g_{11}}X_{1}, \quad X_{12}\react{g_{12}} \emptyset.
\end{array}
\end{equation}
Again, \cite{DelVecchioMurray2014} also proposed a reduced-order model
$$X_1+X_2 \revreact{g_{12}}{g_{3}}  x_3 \react{g^*_3} X_2+ X_4 +X_5,~~~
x_3\react{g_3} X_1,~~~X_4 \react{g_4} \emptyset,$$
for which  a piecewise-linear Lyapunov function in rates is known
to exists \cite{AlrAngSon19}.

We consider the system of differential equations associated with the complete reaction network \eqref{react:transcription}, which is
$$
\left \{
\begin{array}{rcl}
\dot x_1 &=& -g_{1}(x_1) + g_{2}(x_2) +g_{11}(x_{11})\\
\dot x_2 &=&  g_{1}(x_1) -g_{2}(x_2) - g^*_2(x_2) + g_{3}(x_3)  \\
\dot x_3 &=&  g_{2}^*(x_2) -g_{3}(x_3) - g_{34}(x_3,x_4) +g(x_5) \\
\dot x_4 &=&  - g_{34}(x_3,x_4) + g(x_5) + g_{6}(x_6)\\
\dot x_5 &=&   g_{34}(x_3,x_4) - g(x_5) - g_5^*(x_5) \\
\dot x_6 &=&   g_5^*(x_5)  - g_{6}(x_6) \\
\dot x_7 &=&   g_{6}(x_6)-g_{7}(x_7) \\
\dot x_8 &=&   g_{7}(x_7)-g_{8}(x_8) \\
\dot x_9 &=&   g_{8}(x_8)-g_{9}(x_9) \\
\dot x_{10} &=& g_{9}(x_9) - g_{10}(x_{10})\\
\dot x_{11} &=& g_{10}(x_{10}) - g_{11}(x_{11})
\end{array}
\right .
$$
We can neglect the additional equation $\dot x_{12} = g_{10}(x_{10}) - g_{12}(x_{12})$, because, once we prove that the $11$-order system is stable and converges to an equilibrium, convergence of $x_{12}$ to an equilibrium immediately follows. For this system, matrices $B$ and $C$ are
$$B =
\left[
\begin{smallmatrix}
    -1   &  1   &  0   &  0   &  0   &  0   &  0   &  0   &  0   &  0   &  0   &  0   &  0   &  1\\
     1   & -1   & -1   &  1   &  0   &  0   &  0   &  0   &  0   &  0   &  0   &  0   &  0   &  0\\
     0   &  0   &  1   & -1   & -1   & -1   &  1   &  0   &  0   &  0   &  0   &  0   &  0   &  0\\
     0   &  0   &  0   &  0   & -1   & -1   &  1   &  0   &  1   &  0   &  0   &  0   &  0   &  0\\
     0   &  0   &  0   &  0   &  1   &  1   & -1   & -1   &  0   &  0   &  0   &  0   &  0   &  0\\
     0   &  0   &  0   &  0   &  0   &  0   &  0   &  1   & -1   &  0   &  0   &  0   &  0   &  0\\
     0   &  0   &  0   &  0   &  0   &  0   &  0   &  0   &  1   & -1   &  0   &  0   &  0   &  0\\
     0   &  0   &  0   &  0   &  0   &  0   &  0   &  0   &  0   &  1   & -1   &  0   &  0   &  0\\
     0   &  0   &  0   &  0   &  0   &  0   &  0   &  0   &  0   &  0   &  1   & -1   &  0   &  0\\
     0   &  0   &  0   &  0   &  0   &  0   &  0   &  0   &  0   &  0   &  0   &  1   & -1   &  0\\
     0   &  0   &  0   &  0   &  0   &  0   &  0   &  0   &  0   &  0   &  0   &  0   &  1   &  -1
\end{smallmatrix}
\right]$$
and
$$
C =
\left[
\begin{smallmatrix}
     1   &  0   &  0   &  0   &  0   &  0   &  0   &  0   &  0   &  0   &  0\\
     0   &  1   &  0   &  0   &  0   &  0   &  0   &  0   &  0   &  0   &  0\\
     0   &  1   &  0   &  0   &  0   &  0   &  0   &  0   &  0   &  0   &  0\\
     0   &  0   &  1   &  0   &  0   &  0   &  0   &  0   &  0   &  0   &  0\\
     0   &  0   &  1   &  0   &  0   &  0   &  0   &  0   &  0   &  0   &  0\\
     0   &  0   &  0   &  1   &  0   &  0   &  0   &  0   &  0   &  0   &  0\\
     0   &  0   &  0   &  0   &  1   &  0   &  0   &  0   &  0   &  0   &  0\\
     0   &  0   &  0   &  0   &  1   &  0   &  0   &  0   &  0   &  0   &  0\\
     0   &  0   &  0   &  0   &  0   &  1   &  0   &  0   &  0   &  0   &  0\\
     0   &  0   &  0   &  0   &  0   &  0   &  1   &  0   &  0   &  0   &  0\\
     0   &  0   &  0   &  0   &  0   &  0   &  0   &  1   &  0   &  0   &  0\\
     0   &  0   &  0   &  0   &  0   &  0   &  0   &  0   &  1   &  0   &  0\\
     0   &  0   &  0   &  0   &  0   &  0   &  0   &  0   &  0   &  1   &  0\\
     0   &  0   &  0   &  0   &  0   &  0   &  0   &  0   &  0   &  0   &  1
 \end{smallmatrix}
\right].
$$
Note that, given that the corresponding derivatives add up to zero, the following sums of concentrations are constant: $x_1 + x_2 + x_3 + x_5 + x_6 + x_7 + x_8 + x_9 + x_{10} + x_{11} \doteq w_1$ and $x_4 + x_5+ x_6 \doteq w_2$ for all $t$.
Along with $x \geq 0$, this ensures boundedness of the stoichiometric compatibility
class, hence the existence of an equilibrium \cite{BG2021survey}.
Again, we replace $x_1$ and  $x_2$ by  $w_1$ and  $w_2$, for which $\dot w_1=0$ and  $\dot w_2=0$, and the corresponding equations are removed.

For this system, Procedure \ref{proc} does not converge, hence we cannot prove structural stability.
However, according to \cite{DelVecchioMurray2014}, a negative regulatory action can be present, due to a repressor signal acting on variable $x_4$ (the DNA promoter).
To investigate the case in which $x_4$ is under a feedback action,
we can pin node $x_4$: then, Procedure \ref{proc} converges, providing a PLF whose unit ball has $54$ vertices (the function in the reduced space has $52$ vertices).
\end{example}

%

\section{Concluding Discussion}

The main contribution of this paper is twofold. First, we have shown that,
for dynamical systems associated with unitary chemical reaction networks, the existence of a polyhedral Lyapunov function is equivalent to the finiteness of the reachable set of an associated generalised Petri net, with possibly negative token numbers. Second, we have shown that applying a pinning control to some nodes is structurally equivalent to converting the corresponding nodes of the generalised Petri net into \emph{black holes} that swallow any incoming token.

Pinning a node means applying a strong local feedback that keeps the node variable constant. For a biochemical reaction network, enforcing actions that keep the concentration of a species constant seems indeed a viable control approach, which we conjecture is actually used in natural systems to stabilise important cellular processes. In many cases, when a species is far more abundant than all the other chemical species involved in the reaction network, its concentration can be regarded as constant, because it is essentially unchanged by the process, while other concentrations are subject to ample fluctuations: this can already be seen as an ``embedded'' pinning control action.

Future research directions along these lines include considering more general types of structural feedback laws. Another aspect we leave for the future is how to fit this framework in a stochastic setting in which the transitions are probabilistic, going beyond the worst case structural analysis provided here.



{}

\appendix

\section*{Proof of Lemma \ref{sublemma}}
The set ${\cal S}$ is positively invariant for \eqref{systemmu} for any large $\mu>0$.
Consider the modified set
$$
{\cal S}_\nu=\{z \in {\cal S} \colon \|z_1\| \leq \nu \},
$$ 
which corresponds to the portion between the cyan planes in Fig. \ref{fig:lemma}.
For any $\nu>0$ (no matter how small) there exist $\hat \mu$ such that
${\cal S}_\nu$ becomes
positively invariant for $\mu \geq \hat \mu$. 
\begin{figure}[tb]
\centering
\includegraphics[width=5cm]{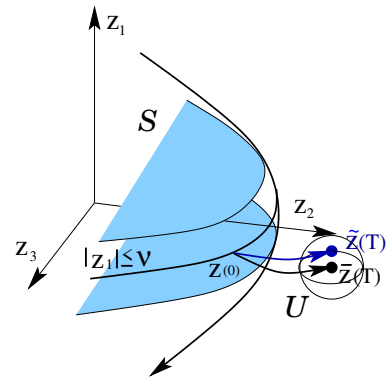}
\caption{Graphical representation of Lemma \ref{sublemma}.}
\label{fig:lemma}
\end{figure}
Indeed,
$$
\dot z_1 = F_{11} z_1- \mu  z_1 + F_{12} z_2,
$$
where both $ F_{12} z_2 $ and $  F_{11} z_1  $ are bounded terms,
because ${\cal S}_\nu$ is a compact set. Therefore, we can write
$| z_1^\top F_{12} z_2| \leq \alpha$ and $|  z_1^\top F_{11} z_1| \leq \beta $
for suitable positive coefficients $\alpha$ and $ \beta $.
Consider the candidate Lyapunov function
$W(z_1) =\frac{1}{2} \|z_1\|^2$ and the ball $\|z_1\|^2 \leq \nu^2$.
Then 
$$
\dot W(z_1) = z_1^\top F_{11}z_1 - \mu z_1^\top z_1 + z_1^\top F_{12} z_1
\leq - \mu z_1^\top z_1 + \alpha + \beta,
$$
hence $\dot W(z_1)<0$ for $\|z_1\|^2 > \nu^2$ provided that
$\mu \geq \hat \mu= (\alpha + \beta)/\nu^2$.

Therefore, assume that $\|z_1\| \leq \nu$, $\mu \geq \hat \mu$.
Take the initial condition $z_1(0)=0$ and $\bar z_2=z_2(0)$ on the boundary
of ${\cal S}_2$.

By contradiction, assume that the solution  $\bar z_2(t)$
of $\dot z_2 = F_{22} z_2$ leaves the set ${\cal S}_2$. Consider this
solution in the extended $z$ space, $\hat z(t)=[0~\bar z_2(t)^\top]^\top$.  
There exists a time instant $T$ such that
$\hat z(T)=[0~\bar z_2(T)^\top]^\top$ is outside the compact ${\cal S}$ and there exists
a neighbourhood ${\cal U}$ (the ball in Fig. \ref{fig:lemma}) centred at $\bar z(T)$, which has no intersection with
 ${\cal S}$. Note that $\hat z(T)$ does not depend on $\mu$.

Now, consider the solution $\tilde z(t) $ of the full system (which depends on $\mu$)
with the same initial condition $\bar z=[0~\bar z_2(T)^\top]^\top$.

We complete the proof by showing that $\tilde z(T)$ gets arbitrarily close to
$\hat z(T) $ if $\mu$ is large enough. The first component satisfies 
$\|\tilde z_1\| \leq \nu$. The second component satisfies
$
\dot{\tilde{z}}_2 = F_{22} \tilde z_2 + F_{21} \tilde z_1.
$
Then, the difference $\tilde z_2(t) - \bar z_2(t)$ satisfies
$$
\frac{d}{dt} [\tilde z_2 - \bar z_2] = F_{22} [\tilde z_2- \bar z_2] + F_{21} \tilde z_1,~~
\| \tilde z_1(t)\| \leq \nu,~~\forall t
$$
with $\tilde z_2(0) - \bar z_2(0)=0$.  Hence
\begin{eqnarray*}
\|\tilde z_2(T) - \bar z_2(T)\| &=& \left \|\int_0^T~e^{F_{22}(T-t)} F_{21} \tilde z_1(t) dt \right \| 
\\
&\leq& \|\nu\| \left \|\int_0^T~e^{F_{22}(T-t)} F_{21}dt \right \|.
\end{eqnarray*}
Given any small $\rho>0$, we may ensure $ \|\tilde z_2(T) - \bar z_2(T)\| \leq \rho$ 
by forcing   a small enough $\nu$. Since both $\nu$ and $\rho$ can be arbitrarily small, we get that $\hat z(T) \in {\cal U}$, hence it is outside ${\cal S}$, against the invariance assumption. We have reached a contradiction, which completes the proof.

\begin{IEEEbiography}[{\includegraphics[width=1in,height=1.25in,clip,keepaspectratio]{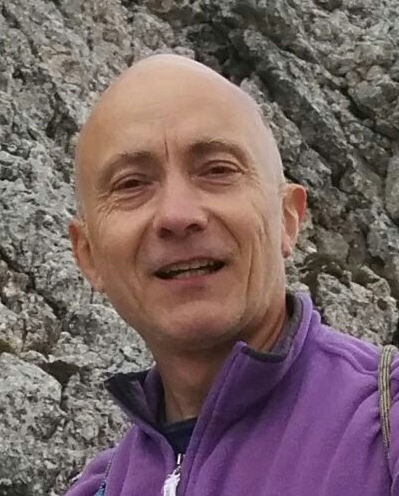}}]{Franco Blanchini} (Senior Member, IEEE) was born on 29 December 1959, in Legnano (Italy). He is the Director of the Laboratory of System Dynamics at the University of Udine. He has been involved in the organization of several international events: in particular, he was Program Vice-Chairman of the conference Joint CDC-ECC 2005, Seville, Spain; Program Vice-Chairman of the Conference CDC 2008, Cancun, Mexico; Program Chairman of the Conference ROCOND, Aalborg, Denmark, June 2012 and Program Vice-Chairman of the Conference CDC 2013, Florence, Italy. He is co-author of the book ``Set theoretic methods in control'', Birkhäuser. He received the 2001 ASME Oil \& Gas Application Committee Best Paper Award as a co-author of the article ``Experimental evaluation of a High-Gain Control for Compressor Surge Instability'', the 2002 IFAC prize Survey Paper Award as the author of the article ``Set Invariance in Control - a survey'', Automatica, November 1999, for which he also received the High Impact Paper Award in 2017, and the 2017 NAHS Best Paper Award as a co-author of the article ``A switched system approach to dynamic race modelling'', Nonlinear Analysis: Hybrid Systems, 2016.
He was nominated Senior Member of the IEEE in 2003. He has been an Associate Editor for Automatica, from 1996 to 2006, and for IEEE Transactions on Automatic Control, from 2012 to 2016. From  2017 to 2019 he has been an Associate Editor for Automatica. He has been a Senior Editor for IEEE Control Systems Letters.
\end{IEEEbiography}

\begin{IEEEbiography}[{\includegraphics[width=1in,height=1.25in,clip,keepaspectratio]{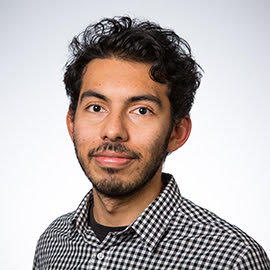}}]{Carlos Andrés Devia} (Student Member, IEEE) received his B.Sc. and M.Sc. in electronic engineering in 2015 and 2018 respectively, from Pontificia Universidad Javeriana, Bogotá, Colombia. From 2017 to 2019 he was an instructor at the same institution. Since 2019 he is a PhD candidate at the Delft Center for Systems and Control, Delft University of Technology, The Netherlands. His research interests include networked systems with emphasis on biological systems and opinion formation models.
\end{IEEEbiography}

\begin{IEEEbiography}[{\includegraphics[width=1in,height=1.25in,clip,keepaspectratio]{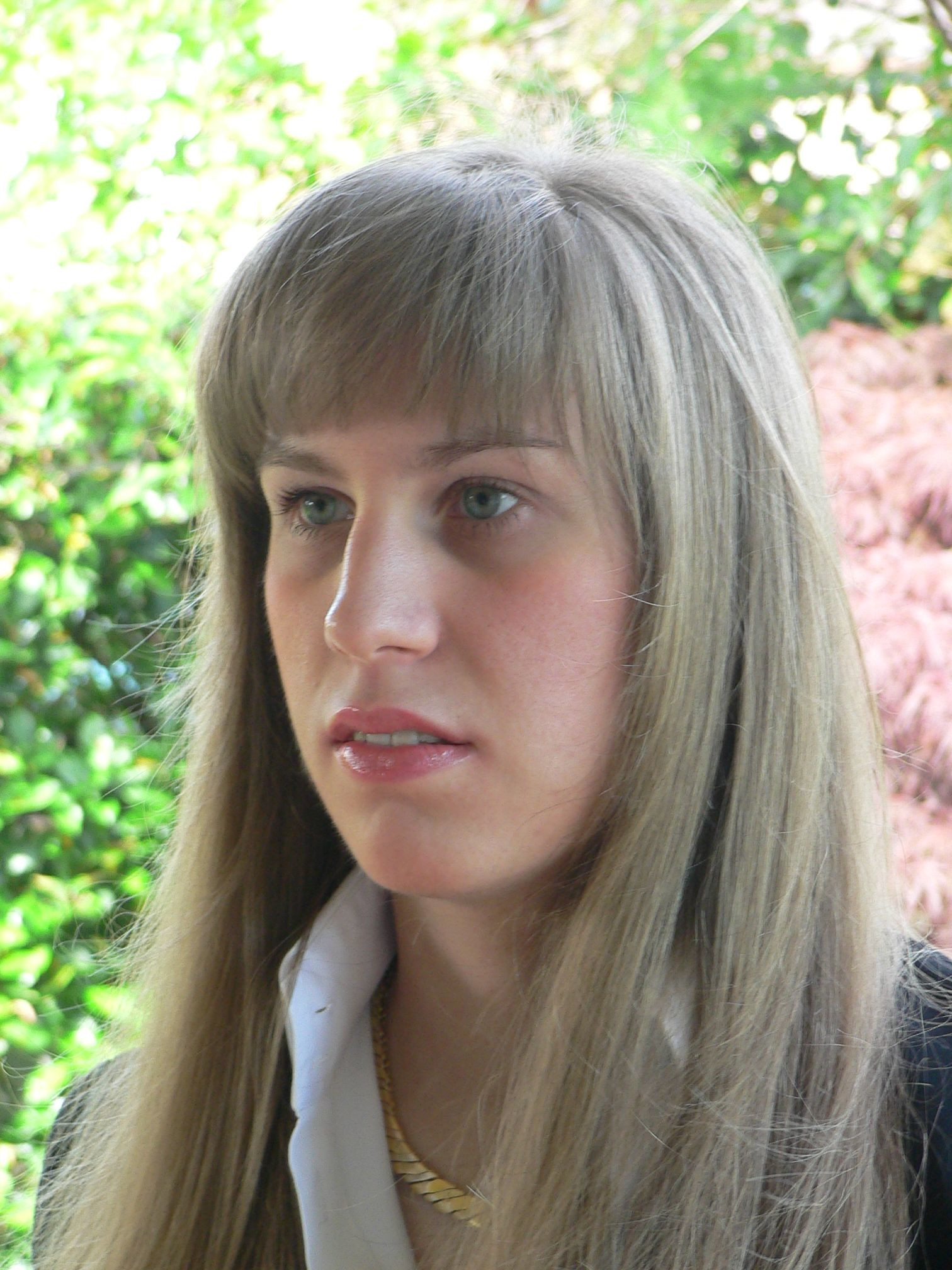}}]{Giulia Giordano} (Member, IEEE) received the B.Sc. and M.Sc. degrees summa cum laude in electrical engineering and the Ph.D. degree (Hons.) in systems and control theory from the University of Udine, Italy, in 2010, 2012, and 2016, respectively. She visited the California Institute of Technology, Pasadena (CA), USA, in 2012, and the University of Stuttgart, Germany, in 2015. She was a Research Fellow at Lund University, Sweden, from 2016 to 2017, and an Assistant Professor at the Delft University of Technology, The Netherlands, from 2017 to 2019. She is currently an Assistant Professor at the University of Trento, Italy. She was recognised with the Outstanding Reviewer Letter from the IEEE Transactions on Automatic Control in 2016 and from the Annals of Internal Medicine in 2020. She received the EECI Ph.D. Award 2016 for her thesis ``Structural Analysis and Control of Dynamical Networks'', the NAHS Best Paper Prize 2017, as a coauthor of the article ``A Switched System Approach to Dynamic Race Modelling'', Nonlinear Analysis: Hybrid Systems, 2016, and the SIAM Activity Group on Control and Systems Theory Prize 2021. Her main research interests include the study of dynamical networks, the analysis of biological systems, and the control of networked systems.
\end{IEEEbiography}

\end{document}